\documentclass[11pt,reqno]{article}   
\usepackage{fullpage}

\usepackage[unicode=true]{hyperref}
\usepackage{amsmath}
\usepackage{cite}
\usepackage{amsfonts}
\usepackage{amssymb}
\usepackage{amsthm}
\usepackage{authblk}
\usepackage{mypack} 
\usepackage[pdftex]{color,graphicx}
\usepackage{adjustbox}
\usepackage{soul}
\usepackage{comment}
\usepackage{bbold}
\usepackage{tikz} 
\usetikzlibrary{decorations.pathreplacing,angles,quotes}
\usetikzlibrary{matrix,arrows,decorations.pathmorphing}
\usepackage{tikz-cd}
\usetikzlibrary{arrows}
\usetikzlibrary{decorations.markings}

\usepackage{bbold}
\usepackage{diagbox}
\usepackage{caption}
\usepackage{subcaption}
\usepackage{pdfpages} 
\usepackage{blkarray}
\usepackage{centernot}
\usepackage{mathtools}
\usepackage{stmaryrd}
\usepackage{soul}  
\usepackage{arydshln}

\definecolor{bluegray}{rgb}{0.4, 0.6, 0.8}
\definecolor{turquoise}{rgb}{0.2, 0.7, 0.6}
\definecolor{pinegreen}{HTML}{008B72}
\newcommand\iis[1]{{\color{olive}#1}}	

\title{Equivariant {simplicial distributions and quantum} contextuality}
 
\author{Cihan Okay\footnote{cihan.okay@bilkent.edu.tr} }
\author{Igor Sikora\footnote{igor.sikora@bilkent.edu.tr}} 
\affil{{\small{Department of Mathematics, Bilkent University, Ankara, Turkey}}}

\date{\today
}

\DeclareMathOperator{\id}{id}
\begin{document}
  \maketitle 
  
\begin{abstract}
We introduce an equivariant version of contextuality with respect to a symmetry group, which comes with natural applications to quantum theory.
In the equivariant setting, we 
{construct} cohomology classes that can detect contextuality. 
This framework is motivated by the earlier topological approach {to contextuality} producing cohomology classes that serve as computational primitives in {measurement-based quantum computing.}
\end{abstract}   
  
\tableofcontents

\section{Introduction}
 
{In quantum theory, measurement statistics is described by a family of probability distributions associated to quantum measurements that can be performed simultaneously.
Such a family of distributions is called contextual if a joint probability distribution cannot describe the measurement statistics.
The sheaf-theoretic approach of \cite{abramsky2011sheaf} is a natural framework to capture this notion.
More recently, a topological approach has been introduced in \cite{Coho} based on chain complexes and tailored for applications in measurement-based quantum computing.
Simplicial distributions introduced in \cite{okay2022simplicial} is a new framework that unifies these two earlier approaches and goes beyond by generalizing the notion of contextuality to  measurements and outcomes described by spaces rather than discrete sets.} 
In \cite{Coho}, topological techniques are first introduced to describe proofs of contextuality for (1) parity proofs based on algebraic relations among the observables and (2) symmetry-based proofs that rely on a transformation {group}  
acting on the observables. Moreover, the two kinds of proofs can be related using the theory of chain {complexes}.
This paper extends this analysis to simplicial distributions.
Our main contribution is the notion of equivariant contextuality with respect to a symmetry group {and the construction of cohomology classes, which can detect this type of contextuality. These cohomology classes are related to those introduced in \cite{Coho}, which play an important role} in quantum computational schemes such as measurement-based quantum computation \cite{raussendorf2016cohomological,Fraction2018} and quantum computation with magic states \cite{raussendorf2023role} for quantifying computational advantage. As a future direction, the authors plan to investigate the computational {significance} of the {cohomology classes} introduced in the equivariant setting.

The framework of simplicial distributions is based on simplicial sets, combinatorial models of spaces more expressive than simplicial complexes.
A simplicial distribution on a pair of simplicial sets $X$ and $Y$, which represent the measurements and outcomes, is given by a simplicial set map
$$
p:X\to D_R(Y)
$$ 
The simplicial set $D_R(Y)$ has simplices given by $R$-valued distributions on the set of simplices of $Y$, where $R$ is a semiring.
Any simplicial set map $s:X\to Y$ can be turned into a simplicial distribution $\delta^s=\delta_Y\circ s$,  called a deterministic distribution, by composing with $\delta_Y:Y\to D_R(Y)$ {that} sends a simplex of the outcome space to the delta-distribution peaked at that simplex.
A simplicial distribution is called contextual if it cannot be expressed as a probabilistic mixture of deterministic distributions. 
Given a group $G$ and {a pair of} simplicial $G$-sets {$(X,Y)$}, we introduce the notion of $G$-equivariant contextuality.
The basic idea is to require equivariance with respect to $G$ in the definitions of simplicial and deterministic distributions.

In quantum theory, commutation relation {among operators} can be used to define simplicial sets that are variations of the well-known nerve construction in algebraic topology. 
Such simplicial sets are called partial groups, {a notion} introduced in \cite{broto2021extension}.
A typical example is defined for the unitary group $U(\hH)$ of a finite dimensional Hilbert space $\hH$ and consists of the $n$-simplices given by
$$
N(\ZZ_d,U(\hH)) = \set{(A_1,A_2,\cdots,A_n):\, A_i^d=\one,\;\;A_iA_j=A_jA_i\; \forall i,j }.
$$
This space and its variants have been studied {in} 
\cite{adem2015classifying,antolin2020classifying,
okay2021classifying}.
Identifying unitary operators that differ by an element in $\Span{e^{2\pi i/d}\one}$ gives a central partial group extension in the sense of \cite{broto2021extension}. 
We develop our theory for a central partial group extension {of the form}
$$
NA \xrightarrow{i} E\xrightarrow{\pi} M
$$ 
{where $A$ is a finite abelian group and $NA$ is the  nerve space.}
{This} extension is classified by a cohomology class {$[\beta]\in H^2(M,A)$}. A simplicial distribution $p:E\to D_R(NA)$ is said to be relative to $i:NA\to E$ if the following diagram commutes
$$
\begin{tikzcd}[column sep=huge,row sep=large]
NA \arrow[d,"i"] \arrow[rd,"\delta_Y"] & \\
X\arrow[r,"p"] & D_R(NA)
\end{tikzcd}
$$ 
We show that if $[\beta]\neq 0$, then every such simplicial distribution is contextual (Corollary \ref{cor:beta nonzero then contextual}).
In the equivariant setting, we consider a group $G$ that acts on $E$ by partial group automorphisms fixing $NA$. Our main tool {in the equivariant setting} is the Borel construction, which sends a simplicial $G$-set $X$ to the simplicial set given by $X\hquo G = EG\times_G X$. {Applying this construction} to the partial group extension gives another partial group extension
$$
NA \to E\hquo G \to M\hquo G
$$ 
Let $[\beta_G]\in H^2(M\hquo G,A)$ {denote} the cohomology class 
{classifying this}
extension.
On the other hand, we can think of the Borel construction as the diagonal of the bisimplicial set $S_G(X)$ whose $(p,q)$-simplices are given by $(EG)_p \times_G X_q$.  
The total complex of the associated double complex of cochains on this bisimplicial set provides an alternative way of studying the cohomology class $[\beta_G]$.
In the total complex, we have two important cocycles: 
First of all, $\beta$ can be regarded as a $(0,2)$-cocycle.
Given a pseudo-section $\eta:M\to E$  of $\pi$ (similar to a simplicial set map except that it fails to be compatible with the $d_0$ face), we have a $(1,1)$-cocycle $\Phi:G\times M_1\to A$ defined by 
$$
\Phi_g(x) = (g\cdot \eta(g^{-1} \cdot x)) \cdot \eta(x)^{-1}.
$$
Our main result is the following (Theorem \ref{thm:betaG zero iff there exists s iff equivariant section} and Corollary \ref{cor:beta-G nonzero then G-equiv contextuality}): 
\begin{thm*} 
Let $NA \xrightarrow{x} E \xrightarrow{\pi} M$ be a central partial group extension and $G$ be a group acting on $E$ {by} partial group {automorphisms}   that fix $NA$.
The following are equivalent:
\begin{enumerate}
\item The class $[\beta_G]$ is zero.  
\item There exists $s:M_1\to A$ such that $d^vs=\beta$ and $d^hs=\Phi$.
\item The map $\pi$ admits a $G$-equivariant section.
\end{enumerate} 
{Moreover, if $[\beta_G]\neq 0$ then every $G$-equivariant distribution $p:E\to D_R(NA)$ relative to $NA$ is $G$-equivariantly contextual.}
\end{thm*}

There is a counterpart to the partial group approach, which uses cofiber sequences of spaces. We develop this approach in Section \ref{sec:equivariant distributions} and prove a result analogous to the theorem above (Theorem \ref{thm:gammaG zero then beta and phi trivial}).
Equivariant distributions naturally arise in quantum theory.  
The nerve space $N(\ZZ_d,U(\hH))$ can be regarded as the total space of a partial group extension and a quantum state (density operator) can be used to obtain a simplicial distribution relative to the fiber, which is $N\ZZ_d$. 

The paper is organized as follows. In Section \ref{sec:equivariant distributions}, we introduce the notion of equivariant distributions and equivariant contextuality. We give a simple example, torus with the involution action, to discuss these basic notions (Section \ref{sec:torus-with-involution}). We also introduce distributions on the simplicial circle.
Section \ref{sec:Distributions relative to a subspace} is about the relative version of contextuality and mainly focuses on introducing the cohomological notions. The total complex of the Borel construction is studied in this section. We end {the} section by elaborating on the torus example.
Quantum distributions are discussed in Section \ref{sec:quantum distributions}. We introduce 
partial group extensions and group actions on them. Our main theorem {stated above} is proved in this section. The counterpart to the torus example is the nerve of the dihedral group, also discussed in this section. We show how quantum theory naturally gives simplicial distributions relative to a subspace. We finish with the 
Mermin star construction, {a well-known example of a contextual scenario} introduced in \cite{mermin1993hidden}, analyzed in the equivariant setting. {Our methods extend} 
the {earlier} topological study {presented} in \cite{Coho} by {demonstrating that this construction produces equivariantly contextual simplicial distributions}.

\paragraph{Acknowledgments.}
This work is supported by the US Air Force Office of Scientific Research under award number FA9550-21-1-0002 and the Digital Horizon Europe project FoQaCiA, GA no. 101070558. 

\section{Equivariant distributions}
\label{sec:equivariant distributions}
 
Simplicial distributions are first introduced in \cite{okay2022simplicial} as a framework to study contextuality. 
In this section we introduce 
equivariant  simplicial distributions {and equivariant contextuality}.

\subsection{Simplicial distributions}
\label{sec:simplicial distributions}

Let $R$ be a commutative zero-sum-free semiring, {i.e., $a+b=0$ implies $a=b=0$ for all $a,b,\in R$.}
The {\it distribution monad}  $D_R:\catSet \to \catSet$ is defined as follows:
	\begin{itemize}
		\item  The set $D_R(U)$ of $R$-distributions on a set $U$ consists of  functions $p:U\to R$ of finite support{, i.e., $|\set{u\in U:\, p(u)\neq 0}|<\infty$,} 
		such that $\sum_{u\in U} p(u)=1$.
		\item Given a function $f:U\to V$ the function $D_R(f):D_R(U)\to D_R(V)$ is defined by
		$$
		p \mapsto \left( v\mapsto \sum_{u\in f^{-1}(v)} p(u) \right). 
		$$  
	\end{itemize} 

Let $\catDelta$ denote the simplex category with objects $[n]=\set{0,1,\cdots,n}$ where $n\geq 0$ and morphisms given by ordinal maps $\theta:[m]\to [n]$. This category is generated by morphisms of the form
\begin{itemize}
\item coface maps $d^i:[n-1]\to [n]$ that skip $i$ in the target, and
\item codegeneracy maps $s^j:[n+1]\to [n]$ that has a double preimage at $j$.
\end{itemize}
A simplicial set is a functor 
$$
X:\catDelta^\op\to \catSet
$$
A simplicial set map $f:X\to Y$ between two simplicial sets is a natural transformation. We will denote the category of simplicial sets by $\catsSet$.
Alternatively, 
a simplicial set is a sequence of sets $X_n$, where $n\geq 0$, together with the face $d_i$ and the degeneracy maps $s_j$ dual to the coface $d^i$ and the codegeneracy $s^j$ maps.
In this description a simplicial set map consists of a family of functions $\set{f_n:X_n\to Y_n}_{{\geq 0}}$ 
compatible with the face and degeneracy maps. Given $\sigma\in X_n$ we will write 
$f_\sigma\in Y_n$ for the value of the function $f$ on $\sigma$.
The standard $d$-simplex $\Delta[d]$ is the simplicial set with $n$-simplices given by the set of ordinal maps $[n]\to [d]$. In other words, this is the functor represented by $[d]$.
We will write $\sigma^{01\cdots d}$ for the $d$-simplex in $\Delta[d]$ corresponding to the identity map $\idy:[d]\to [d]$. Observe that every other $n$-simplex can be obtained from $\sigma^{01\cdots n}$ by an application of face and degeneracy maps.

The functor $D_R$ extends to a monad on simplicial sets 
$$
D_R:\catsSet \to \catsSet
$$	
sending a simplicial set $X:\catDelta^\op \to \catSet$ to the composite functor $D_R(X):\catDelta^\op \xrightarrow{X} \catSet \xrightarrow{D_R} \catSet$; see \cite{kharoof2022simplicial}.
We will write $\delta_X:X\to D_R(X)$ for the unit of this monad {that sends $x$ to the delta-distribution with a peak at $x$}. 

\Def{\label{def:simplicial distributions}
Let $X$ and $Y$ be simplicial sets.
A \emph{simplicial distribution} on $(X,Y)$ is a simplicial set map $p:X\to D_R(Y)$. 
A \emph{deterministic distribution} is a simplicial distribution of the form $\delta^s:X\xrightarrow{s} Y \xrightarrow{\delta_Y} D_R(Y)$ where $s:X\to Y$ is a simplicial set map. 
We write $\sDist(X,Y)$ and $\dDist(X,Y)$ for the sets of simplicial and deterministic distributions, respectively.
There is a comparison map
$$
\Theta: D_R(\dDist(X,Y)) \to \sDist(X,Y)
$$
that sends a distribution $d=\sum_r d(r)\,\delta^r$ 
to the simplicial distribution $\Theta(d):X\to D_RY$  defined by 
$$
\Theta(d)_\sigma: \theta \mapsto \sum_{r\,:\,r_\sigma=\theta} d(r) 
$$ 
where $\sigma\in X_n$, $\theta\in Y_n$ and $r:X\to Y$ runs over {simplicial set maps} 
such that $r_\sigma=\theta$.
A simplicial distribution is called {\it non-contextual} if it lies in the image of $\Theta$. Otherwise, it is called {\it contextual}. 
}

{

\Ex{\label{ex:simp-dist}
{\rm
Let us consider simplicial distributions of the form $p:\Delta[2]\to D_R(N\ZZ_2)$. Such a simplicial set map is determined by a distribution $p_\sigma\in D_R(\ZZ_2^2)$ where $\sigma=\sigma^{012}$.
Writing $p^{ab}$ for the values $p_\sigma(a,b)$ this distribution is specified by a tuple $(p^{00}, p^{01}, p^{10}, p^{11})\in R^4$ such that $\sum_{a,b}p^{ab}=1$.
A deterministic distribution $\delta^s$ for a simplicial set map $s:\Delta[2]\to N\ZZ_2$ will be determined by a delta-distribution at $s_\sigma=(a,b)\in \ZZ_2^2$. We write $\delta^{ab}$ for this simplicial distribution. Then any simplicial distribution $p$ can be written as
$$ 
p = \sum_{a,b} p^{ab} \delta^{ab},
$$
and hence is non-contextual. More complicated simplicial distributions can be constructed by gluing $2$-simplices. In this case additional relations will be introduced as the faces. The distribution $p_{d_i\sigma}$ associated to a face of $\sigma$ will be an element of $D_R(\ZZ_2)$. We will write $(p_{d_i\sigma})^a$ to denote its value as $a\in \ZZ_2$. We have $(p_{d_i\sigma})^0+(p_{d_i\sigma})^1=1$. Then
$$
(p_{d_i\sigma})^0 = \left\lbrace
\begin{array}{ll}
p^{00}+p^{10}  & i=0 \\
p^{00}+p^{11}  & i=1 \\
p^{00}+p^{01}  & i=2. 
\end{array}
\right.
$$
When two $2$-simplices are glued at a common face, say the $i$-th face of $\sigma_1$ and the $j$-th face of $\sigma_2$, then we have
$$
d_ip_{\sigma_1} = p_{d_i\sigma_1} = p_{d_j\sigma_2} = d_jp_{\sigma_2}.
$$  
This way we can construct contextual simplicial distributions. See \cite{okay2022simplicial} and \cite{kharoof2023topological} for examples.
}}

}


\subsection{Partial groups}


In this section we recall the notion of partial monoids and groups from \cite{broto2021extension}.

We need two ordinal maps:
\begin{itemize}
\item The $i$-th edge map $e^i:[1]\to [n]$ that sends $\set{0,1}$ to $\set{i-1,i}$.
\item The $n$-th multiplication map $\Pi^n:[1]\to [n]$  that sends $\set{0,1}$ to $\set{0,n}$. 
\end{itemize}  
 
\Def{\label{def:partial monoid}
A {\it partial monoid} is a non-empty simplicial set $M$ such that
\begin{itemize}
\item $M$ is reduced, i.e., $M_0=\set{\ast}$,
\item The $n$-th spine map 
$$
\ee_n=(e_1,\cdots,e_n): M_n\to \underbrace{M_1\times \cdots \times M_1}_{n}
$$
is injective for all $n\geq 1${, where $e_i:M_n\to M_1$ is induced by the ordinal map $e^i$.}
\end{itemize} 
A homomorphism of partial monoids is a simplicial set map. 
} 

A simplex $x\in M_n$ can be identified with the tuple $\ee_n(x)=(x_1,\cdots,x_n)$. 
{Let $\Pi_n:M_n\to M_1$ denote the map induced by the ordinal map $\Pi^n$.}
The simplex $\Pi_n(x)\in M_1$ will also be denoted by $x_1\cdot x_2  \cdots  x_n$ to emphasize that this element represents the product. 
The degenerate simplex $s_0(\ast)\in M_1$ will be denoted by $1$ since it serves as the identity of the product.

 \begin{lem}
 [\!\cite{broto2021extension}]
\label{lem:maps into partial monoids and homotopy} 
Let $X$ be a simplicial set, and $M$, $N$ be a partial monoids.
\begin{enumerate}
\item A simplicial set map $f:X\to M$ is determined by $f_1:X_1\to M_1$.
\item A simplicial homotopy $F:N\times \Delta[1] \to M$ from $g$ to $f$ is determined by $\theta=F_1(1,\sigma^{01})\in M_1$ that satisfies $\theta\cdot g_1(x) = f_1(x)\cdot \theta$ for all $x\in N_1$.
\end{enumerate} 
\end{lem}

{Part (1) of this lemma justifies saying that $f:X\to M$  is defined by $f(x)=m$ on $1$-simplices. Note that we also drop the subscript from $f_1$.}
We write $f\xrightarrow{\theta} g$ for the homotopy from $f$ to $g$.
Every monoid can be regarded as a partial monoid by the nerve construction.

\Def{ \label{def:nerve space}
For a monoid $M$ the {\it nerve space} $N(M)$  is the simplicial set whose set of $n$-simplices is given by $M^n$ with the following simplicial structure:
$$
\begin{aligned}
d_i(m_1,m_2,\cdots,m_n) &= \left\lbrace \begin{array}{ll}
(m_2,m_3,\cdots,m_n)   &  i=0\\
(m_1,\cdots,m_i\cdot m_{i+1} ,\cdots,m_n) & 0<i<n\\
(m_1,m_2,\cdots,m_{n-1}) & i=n
\end{array}
\right. \\
s_j(m_1,m_2,\cdots,m_n) &= (m_1,\cdots,{m_{j}, 1_M,m_{j+1}},\cdots,  m_n)\;\;\, 0\leq j\leq n,
\end{aligned}
$$
where $1_M$ is the identity element.
}

\Cor{\label{cor:maps to NM}
A simplicial set map $f:X\to NM$ is given by a function $f_1:X_1\to M$ such that 
$$
f_1(d_1\sigma) = f_1(d_2\sigma)\cdot f_1(d_0\sigma)
$$
for every $\sigma\in X_2$.
} 
\Proof{{This can be obtained from Lemma \ref{lem:maps into partial monoids and homotopy} by working out the simplicial structure. A more direct approach is to use} 
the fact that the nerve of a category is $2$-coskeletal {\cite[Cor. 1.2]{joyal2008theory}}.
}

\Def{\label{def:partial group}
An inversion in a partial monoid $M$ is a simplicial set map $\nu:M\to M^\op$ such that for every $x\in M_n$ with $n\geq 1$ there is a simplex $(\nu(x),x)\in M_{2n}$ such that $\Pi_{2n}(\nu(x),x)=1$.
A partial monoid with an inversion is called a \emph{partial group}. 
}

A group $G$ can be regarded as a partial group by the nerve construction.

\subsection{Simplicial $G$-sets}

By a simplicial $G$-set we mean a simplicial object in the category of $G$-sets.

We define a simplicial $G$-set $EG$ whose $n$-simplices are given by  $(EG)_n=G^{n+1}$ and the simplicial structure maps are as follows:
$$ 
d_i(g_0,g_1,\cdots,g_n) = \left\lbrace
\begin{array}{cc}
(g_0,\cdots,g_ig_{i+1},\cdots,g_n)   & 0\leq i {<} n \\
(g_0,g_1,\cdots,g_{n-1})     &   i=n
\end{array}
\right. 
$$
and $s_j(g_0,g_1,\cdots,g_n) = (g_0,\cdots,g_{j},1,g_{j+1},\cdots,g_n)$ {where $0
\leq j\leq n$.}
The group action is multiplication from the left on the first coordinate in each level. The quotient map {gives} the universal principal $G$-bundle:
$$
G\to EG\to NG
$$

\Def{\label{def:Borel construction}
Let $X$ be a simplicial $G$-set. 
The {\it Borel construction} of $X$ is the simplicial set $X\hquo G$, also denoted by $EG\times_G X$, whose $n$-simplices are given by $EG_n\times_G X_n$ with the diagonal simplicial structure maps. 
}

Let $f:X\to Y$ be a simplicial $G$-map.
There is an associated map between the Borel constructions that makes the following diagram commute:
$$
\begin{tikzcd}[column sep =huge, row sep=large]
EG\times X \arrow[r,"\idy\times f"] \arrow[d] & EG\times Y \arrow[d] \\
X\hquo G \arrow[r,"f_G"] & Y\hquo G  
\end{tikzcd}
$$
where the vertical maps are the quotients under the $G$-action.
When $G$ acts trivially on $X$, the Borel construction can be identified with $X\hquo G = NG\times X$.
Assume that the action of $G$ on $Y$ is trivial.
We will consider the composition 
\begin{equation}\label{eq:Borel construction and projection}
\tilde f_G: X \hquo G \xrightarrow{f_G} Y \hquo G =  NG \times Y \xrightarrow{\pi_2} Y 
\end{equation} 

Applying the Borel construction to $X\to \ast$ gives a fibration sequence
$$
X {\xrightarrow{\iota_X}} X\hquo G \to NG
$$ 
where $\ast\hquo G$ is identified with $NG$.
The first map in the sequence is the canonical inclusion that sends $x\in X_n$ to the simplex $[(1,\cdots,1),x]$. Note that this is a simplicial $G$-map, where the target has the trivial action.

\subsection{Equivariant distributions}

Let $U$ be a $G$-set. We equip $D_R(U)$ with the following $G$-action:
$$
g\cdot p(u) = p(g^{-1}\cdot u).
$$
{Here the action is induced by the standard $G$-action on the set of all functions $U\to R$, {where the action on the target is trivial}.}
This extends to simplicial $G$-sets. That is, if $Y$ is simplicial $G$-set then $D_R(Y)$ is a simplicial $G$-set with the action above in each degree.

\Def{\label{def:equivariant contextuality}
Let $X$ and $Y$ be simplicial $G$-sets.
A {\it $G$-equivariant simplicial distribution} on $(X,Y)$ is a simplicial $G$-map
$$
p: X\to D_R(Y)
$$
A {\it $G$-equivariant deterministic distribution} is a simplicial distribution of the form $\delta^s: X\xrightarrow{s} Y \xrightarrow{\delta} D_R(Y)$ where $s:X\to Y$ is a simplicial $G$-map.
We will write $\sDist_G(X,Y)$ and $\dDist_G(X,Y)$ 
for the sets of $G$-equivariant simplicial and deterministic distributions, respectively. A $G$-equivariant simplicial distribution $p$ is called  {\it $G$-equivariantly contextual} if it does not lie in the image of
$$
\Theta_G: D_R(\dDist_G(X,Y)) \to \sDist_G(X,Y)
$$
Otherwise, it is called {\it $G$-equivariantly non-contextual}. 
}

We will be interested in the case where the $G$-action on $Y$ is trivial. In this case we can apply the Borel construction, as described in (\ref{eq:Borel construction and projection}), to a $G$-equivariant distribution $p:X\to D_R(Y)$ and obtain a simplicial distribution
\begin{equation}\label{dia:tilde p_G construction}
\tilde p_G: X\hquo G \to D_R(Y)
\end{equation}
This construction can be applied to a simplicial  $G$-map $s:X\to Y$ and the corresponding equivariant distribution $\delta^s:X\to D_R(Y)$.

\Lem{\label{lem:commutativity property of Borel construction} 
We have $\widetilde {\delta_G^s} = \delta^{\tilde s_G}$.
}
\Proof{
We want to show that $\tilde \delta_G^s$ factors as $\delta\circ \tilde s_G$. By construction we have
$$
\tilde \delta_G^s[(1,g_1,\cdots,g_n),x] = \delta^{s(x)} = \delta^{\tilde s[(1,g_1,\cdots,g_n),x]}.
$$
}

The construction in   (\ref{eq:Borel construction and projection}) gives a commutative diagram
\begin{equation}\label{dia:commutative diagram equivariant vs Borel}
\begin{tikzcd}[column sep=huge,row sep =large]
D_R(\dDist_G(X,Y)) \arrow[d] \arrow[r,"\Theta_G"] & \sDist_G(X,Y) \arrow[d] \\
D_R(\dDist(X\hquo G,Y)) \arrow[r,"\Theta"] \arrow[u,bend left,dashed] & \sDist(X\hquo G, Y)  \arrow[u,bend right,dashed]
\end{tikzcd}
\end{equation}
where the left vertical map sends $\delta^s$ to $\delta^{\tilde s_G}$, which is also equal to $\widetilde {\delta_G^s}$ by Lemma \ref{lem:commutativity property of Borel construction}, and the right vertical map is defined by $p\mapsto \tilde p_G$.
The dashed arrows are simply obtained by restricting along the inclusion map $X\to X\hquo G$. The composition of the vertical downward map with the upward map is the identity.

\Pro{\label{pro:Equivariant contextual iff Borel contextual}
A $G$-equivariant simplicial distribution $p$ is $G$-equivariantly contextual if and only if $\tilde p_G$ is contextual.
}
\Proof{
Assume that $p$ is $G$-equivariantly non-contextual. Then by the commutativity of Diagram (\ref{dia:commutative diagram equivariant vs Borel}), with the downward vertical arrows, we see that $\tilde p_G$ is non-contextual. Conversely, if $\tilde p_G$ is non-contextual then the same diagram with the upward vertical arrows this time shows that $p$ is non-contextual. 
}

\subsection{Distributions on the circle}


The set of $n$-simplices of the simplicial interval $\Delta[1]$ is given by 
$$
\Delta[1]_n = \set{\theta^{i}|\; 0\leq i\leq n+1}
$$
where $\theta^i:[n]\to [1]$ is the ordinal map such that $|(\theta^i)^{-1}(0)|=i$. The simplicial structure maps are given by
$$
d_j(\theta^i) = \left\lbrace
\begin{array}{ll}
\theta^{i-1} &  j < i \\
\theta^i  & i\leq j 
\end{array}
\right.
\;\;\;\text{ and }
\;\;\;
s_j(\theta^i) = \left\lbrace
\begin{array}{ll}
\theta^{i+1} & j<  i \\
\theta^{i} & i\leq  j. 
\end{array}
\right.
$$
The boundary $\partial \Delta[1]$ consists of the simplicial subset   given by $\set{\theta^0,\theta^{n+1}}$ in each level. The simplicial structure maps are just the identity maps. The simplicial circle $S^1$ is defined to be the simplicial set obtained from $\Delta[1]$ by collapsing its boundary $\partial \Delta[1]$. So $\theta^0$ and $\theta^{n+1}$ are identified to a single point, which we denote by $\star$.
Therefore the set of $n$-simplices is given by
$$
(S^1)_n = \left\lbrace
\begin{array}{ll}
\set{\star} & n=0 \\
\set{\star, \theta^1,\cdots,\theta^n} & n\geq 1.
\end{array}
\right.  
$$ 
The simplicial structure maps on $\theta^i$ are given by
$$
d_j(\theta^i) = \left\lbrace
\begin{array}{ll}
\theta^{i-1} & j< i{\text{ and }}  1<i \\ 
\theta^{i} & i\leq j {\text{ and }}  i<n \\
\star & \text{otherwise}
\end{array}
\right.
\;\;\;\text{ and }
\;\;\;
s_j(\theta^i) = \left\lbrace
\begin{array}{ll}
\theta^{i+1} & j< i \\ 
\theta^{i} & i\leq j.\\
\end{array}
\right.
$$ 
Any simplicial structure map sends $\star$ to itself. 

{A zero-sum-free semiring $R$ admits a partial order: For $a,b\in R$ we write $a\leq b$ if there exists $c\in R$ such that $a+c=b$. In the next result we use this partial order.}

\Pro{\label{pro:distributions on S1}
There is an isomorphism of simplicial sets
$$
D_R(S^1) \to N_{\leq 1}(R)
$$
where $ N_{\leq 1}(R)$ is the simplicial subset of the nerve space $N(R)$ of $(R,+)$, consisting of tuples whose sum is at most $1$. 
}
\Proof{
An $R$-distribution on $(S^1)_n$ is specified by a tuple $(p_1,p_2,\cdots,p_{n})$ where $p_i\in R$ and $\sum_i p_i \leq 1$. The value assigned to $\theta^i$ is given by $p_{i}$ and the value assigned to $\star$ is determined by this tuple, which is given by $1-\sum_i p_i$. The simplicial structure of $S^1$  can be used to show that
$$
d_j(p_1,p_2,\cdots,p_{n}) = \left\lbrace
\begin{array}{ll}
(p_2,\cdots,p_{n}) & j=0\\ 
(p_2,\cdots,p_j+p_{j+1},\cdots, p_{n}) & 0<j<n \\
(p_2,\cdots,p_{n-1}) & j=n
\end{array}
\right.
$$
and
$$
s_j(p_1,p_2,\cdots,p_{n}) = (p_1,\cdots,p_j,0,p_{j+1},\cdots,p_n).
$$ 
}

Note that $N_{\leq 1}(R)$ is a partial monoid {contained in $NR$} therefore combining with Corollary \ref{cor:maps to NM} we obtain:

\Cor{\label{cor:restriction to 2-skeleton for DS1}
Let $i_2:X^{(2)}\to X$ denote the inclusion of the $2$-skeleton. Then the following induced map is an isomorphism
$$
(i_2)^*:\sDist(X,S^1) \to \sDist(X^{(2)},S^1)
$$
}
\begin{proof}
{The inclusions $i_2$ and $D_RS^1\cong N_{\leq 1}R\to NR$ induce a commutative diagram}
$$
\begin{tikzcd}[column sep=huge,row sep =large]
\catsSet(X,D_RS^1) \arrow[d,hook] \arrow[r,"(i_2)_*"] & \catsSet(X^{(2)},D_R S^1) \arrow[d,hook] \\
\catsSet(X,NR) \arrow[r,"\cong"] & \catsSet(X^{(2)},NR)
\end{tikzcd}
$$
{Then $(i_2)_*$ is injective by commutativity of the diagram. For surjectivity we observe that in the commutative diagram}
$$
\begin{tikzcd}[column sep=huge,row sep =large]
X^{(2)} \arrow[dr] \arrow[r,"i_2",hook] \arrow[d,"p"] & X\arrow[d,dashed,"\tilde p"] \\
D_RS^1 \arrow[r,hook] & NR 
\end{tikzcd}
$$
{the lift $\tilde p$ of the diagonal map has image contained in $D_RS^1$ since it is determined by the restriction to the $1$-skeleton of $X$ by Lemma \ref{lem:maps into partial monoids and homotopy}.}
\end{proof}

In particular, for $R=\NN$ we obtain that 
$$
S^1 = D_\NN(S^1) \xrightarrow{\cong} N_{\leq 1}(\NN)
$$
Therefore we also have
$$
(i_2)^*:\dDist(X,S^1) \to \dDist(X^{(2)},S^1)
$$


\subsection{{Example:} Torus with involution}
\label{sec:torus-with-involution}

We will take $X=S^1\times S^1$ and $Y=S^1$.
We will write $\sigma_0$ and $\sigma_1$ for the non-degenerate $2$-simplices of $X$. 
The face maps for $c\in\{0,1\}$ are given by
$$
d_i \sigma_c = \left\lbrace
\begin{array}{ll}
x_{c+1}  & i=0 \\
x   & i=1\\
x_c   & i=2.
\end{array}
\right.
$$
where the summation in the first formula is taken modulo $2$.
The unique vertex {of $X$} will be denoted by $v$.
{We will regard the circle as a simplicial subset of the nerve space
$$
S^1 \to N\ZZ_2
$$
defined by sending $\sigma^{01}\mapsto 1$ in degree $1$. Using the notation of Example \ref{ex:simp-dist} a simplicial distribution $p:\Delta[2]\to D(S^1)$ is specified by a tuple $(p^{00},p^{01},p^{10},p^{11})$ satisfying $\sum_{a,b}p^{ab}=1$ and $p^{11}=0$. Note that $(S_1)_2$ is in bijective correspondence with $\set{(0,0),(0,1),(1,0)}\subset \ZZ_2^2$ hence a distribution $p$ in $D((S^1)_2)$ regarded as a distribution in $D(\ZZ_2^2)$ will have {$p^{11}=0$}. Then} simplicial distributions on $(X,Y)$ are given as follows
$$
\sDist(X,Y) =\set{(t_1,t_2)\in [0,1]^2:\, 1-(t_1+t_2)\geq 0};
$$
see Figure (\ref{fig:simp dist torus}).

\begin{figure}[h!]
\centering
\includegraphics[width=.4\linewidth]{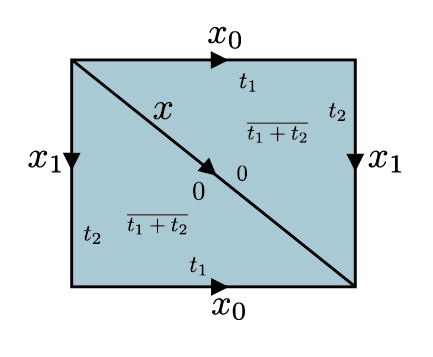}
\caption{Simplicial distribution on the torus consisting of two triangles. Edges labeled by $x_0$ are identified. Similarly for $x_1$.
Top-right triangle is $\sigma_0$ and the bottom-left triangle is $\sigma_1$.  We write $\overline{t_1+t_2}=1-(t_1+t_2)$.}
\label{fig:simp dist torus}
\end{figure} 


Let us take $G=\ZZ_2$ and {endow} 
both spaces with a $G$-action. On $X$ the action is given by the swapping of the coordinates and on $Y$ the action is trivial.
A $G$-equivariant distribution would be one in which $t_1=t_2$. Then  $\sDist_G(X,Y)$ can be identified with the interval $[0,1/2]$. When $t_1=0$ we obtain the unique deterministic distribution $p_{\sigma_i}=\delta^{00}$, 
and for $t_1=1/2$ we obtain 
$$p^{ab}_{\sigma_i} = \left\lbrace
\begin{array}{ll}
1/2 & a+b=1 \\
0 & \text{otherwise.}
\end{array}
\right.
$$  
Therefore $p$ is $G$-equivariantly contextual if and only if $t_1>0$.


\begin{figure}[h!]
\centering
\begin{subfigure}{.30\textwidth}
  \centering
  \includegraphics[width=1\linewidth]{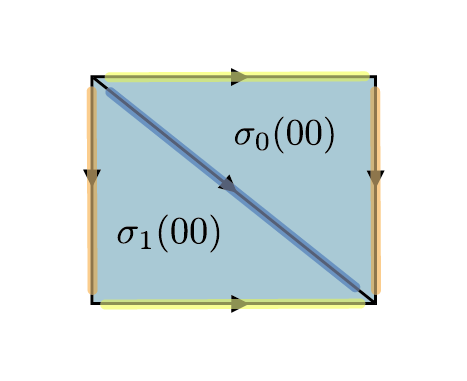}
  \caption{}
  \label{fig:borel1}
\end{subfigure}%
\begin{subfigure}{.45\textwidth}
  \centering
  \includegraphics[width=1\linewidth]{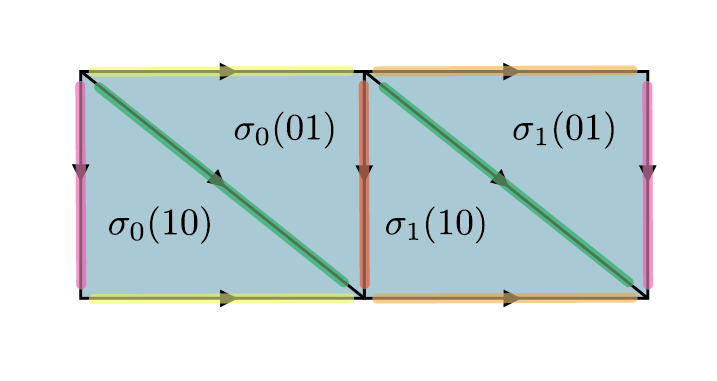}
  \caption{}
  \label{fig:borel2}
\end{subfigure}%
\begin{subfigure}{.25\textwidth}
  \centering
  \includegraphics[width=.8\linewidth]{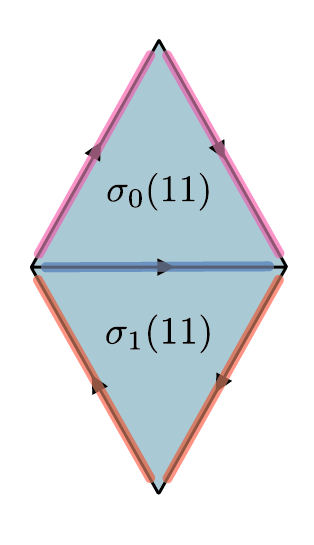}
  \caption{}
  \label{fig:borel3}
\end{subfigure}
\caption{{The non-degenerate $2$-simplices of the Borel construction and the face identifications. {Edges depicted by the same color are identified, e.g., in (a) the top (left) and bottom (right) edges are identified to give a torus.}}}
\label{fig:borel}
\end{figure}

Now, let us consider the Borel construction $X\hquo G$. By Corollary \ref{cor:restriction to 2-skeleton for DS1} it suffices to consider
the $2$-skeleton $W = (X\hquo G)^{(2)}$.
We will use the following notation for the simplices:
$$
v= [(1),v] \;\;\;\;\;\; x_b(a) = [(1,a),x_b] \;\;\;\;\;\; \sigma_c(ab) = [(1,a,b),\sigma_c].
$$
The face maps of the $2$-simplices are given by
$$
d_i \sigma_c(ab) = \left\lbrace
\begin{array}{ll}
x_{a+c+1}(b)  & i=0 \\
x(a+b)   & i=1\\
x_c(a)   & i=2.
\end{array}
\right.
$$
The non-degenerate $2$-simplices can be grouped into three parts each of which is a closed surface {as in Figure (\ref{fig:borel}).}
Let us describe simplicial distributions on $(W,S^1)$. 
{Using}
the 
closed 
surface {in Figure (\ref{fig:borel3})} we have
$$
\begin{aligned}
p_{\sigma_0(11)}^{01}&= p_{\sigma_0(11)}^{10} = t_1\\
p_{\sigma_1(11)}^{01}&= p_{\sigma_1(11)}^{10} = t_2
\end{aligned}
$$  
where $t_1,t_2\in [0,1/2]$. On the other hand, the middle surface forces $t_1=t_2$. Considering the 
surface in {Figure (\ref{fig:borel1})} as well and using Corollary \ref{cor:restriction to 2-skeleton for DS1} we obtain   
{$$
\sDist_G(X,Y) = \set{(t,t):\, t\in [0,1/2]} \subset \sDist(X,Y)
$$
and}
\begin{equation}\label{dia:iso for torus}
\sDist_G(X,Y) \xrightarrow{\cong} \sDist(X\hquo G,Y)
\end{equation}
A similar {isomorphism holds between the} 
deterministic distributions. Hence we verified Proposition 
\ref{pro:Equivariant contextual iff Borel contextual}. 

\section{Distributions relative to a subspace}
\label{sec:Distributions relative to a subspace}



 
In this section $X$ and $Y$ denote simplicial $G$-sets where the $G$-action on $Y$ is trivial. {We consider simplicial distributions relative to a subspace and introduce cohomology classes that can detect contextuality for such distributions.}

\begin{defn}\label{def:G-equivariant relative contextuality}
Let $i:Y\to X$ be a $G$-equivariant inclusion. 
A $G$-equivariant   distribution on $(X,Y)$ {\it relative} to $Y$ is a simplicial $G$-map $p:X\to D_R(Y)$ such that the following diagram commutes
\begin{equation}\label{dia:relative version}
\begin{tikzcd}[column sep = huge, row sep = large]
Y \arrow[dr,"\delta_Y"] \arrow[d,"i"'] & \\
X \arrow[r,"p"] & D_R(Y)
\end{tikzcd}
\end{equation}
\end{defn}

Note that the non-equivariant definition can be obtained by taking $G$ as the trivial group.

\Lem{\label{lem:Borel preserve relative}
If $p:X\to D(Y)$ is a $G$-equivariant   distribution relative to $Y$ then $\tilde p_G:X\hquo G \to D_R(Y)$ {in (\ref{dia:tilde p_G construction})} is a simplicial distribution relative to $Y$.} 
\begin{proof}
{Follows from the following commutative diagram}
\[
\begin{tikzcd}[column sep=huge,row sep=large,ampersand replacement=\&]
Y\ar[d,"i"]\ar[dr,"\delta"]\& \\
X\ar[r,"p"]\ar[d,"\iota_X"]\&D(Y)\\
X\hquo G\ar[ur,"\tilde{p}_G"']\&
\end{tikzcd}
\]
{where the top triangle commutes by assumption on $p$ and the lower triangle commutes by the construction of $\tilde p_G$.}
\end{proof}


We will consider the diagram of cofiber sequences
\begin{equation}\label{dia:cofibrations compare}
\begin{tikzcd}[column sep=huge,row sep =large]
Y \arrow[d,"i"] \arrow[r,equal] & Y \arrow[d,"j"] \arrow[r] & Y\hquo G \arrow[d,"i_G"] \\
X \arrow[d,"q"] \arrow[r,"\iota"] & X\hquo G \arrow[d] \arrow[r,equal] & X\hquo G \arrow[d] \\
\bar X \arrow[r,"\bar\iota"] & \overline{X\hquo G} \arrow[r,"c"]  & \bar X\hquo G
\end{tikzcd}
\end{equation}

\Lem{\label{lem:equivariant section vs normal section}
The following are equivalent.
\begin{enumerate}
\item The inclusion $i:Y\to X$ has a $G$-equivariant section.
\item The inclusion $i_G:Y\hquo G\to X\hquo G$ has a section.
\item The inclusion $j:Y\to Y\hquo G \xrightarrow{i_G}  X\hquo G$ has a section.
\end{enumerate}
} 
\Proof{First we prove the equivalence of (1) and (3).
Let $s:X\to Y$ be a $G$-equivariant section of $i$. The associated map $\tilde s_G:X\hquo G\to Y$ gives a section of $j$. Conversely, let $r:X\hquo G\to Y$ be a section of $j$. Then precomposing with the canonical ($G$-equivariant) map $X\to X\hquo G$ gives a $G$-equivariant section of $i$.  
Next we prove the equivalence of (2) and (3). Let $t$ be a section of $i_G$. Then projecting onto the second factor gives a section of $j$. Given a section $r$ of $j$, composing this with the canonical inclusion $Y\to Y\hquo G$ gives a section of $i_G$.  
}

\subsection{Cohomology}

Our default space of choice for $Y$ will be the nerve space $NA$ {of} some finite abelian group $A$. In this case we can interpret deterministic distributions as cohomology classes. 
Let $Z$ be a simplicial subset of $X$. 
Consider the cofiber sequence 
$$
Z \xrightarrow{i} X \xrightarrow{q} \bar X
$$
We have a commutative diagram
\begin{equation}\label{dia:deterministic dist and cohomology exact}
\begin{tikzcd}[column sep=huge,row sep =large]
\dDist(\bar X,NA) \arrow[d,"\cong"] \arrow[r,"q^*"] & \dDist(X,NA) \arrow[d,"\cong"] \arrow[r,"i^*"] & \dDist(Z,NA) \arrow[d,"\cong"] & \\
H^1(\bar X,A) \arrow[r,"q^*"] & H^1(X,A) \arrow[r,"i^*"] & H^1(Z,A) \arrow[r,"\zeta"] & H^2(\bar X,A) 
\end{tikzcd}
\end{equation}
where the vertical maps send a deterministic distribution $\delta^s$ to the cohomology class $[s]$. The bottom sequence is exact; see {\cite[Sec. 5.3]{okay2022simplicial}} for details.

\Lem{\label{lem:extend iff zero in cohomology}
{Let $Z$ be a simplicial subset of $X$.} 
A simplicial set map $r:Z\to NA$ extends to $X$ if and only if $\zeta([r])=0$.
}
\Proof{ 
From the exactness of the bottom row in Diagram (\ref{dia:deterministic dist and cohomology exact}) we see that $\zeta([r])=0$ if and only if there exists $[\tilde r]\in H^2(X,A)$ such that $i^*([\tilde r])=[r]$. Vertical isomorphisms {imply} that this latter condition is equivalent to the existence of an extension.
}

The cohomology class $\zeta([r])$ is represented by the cocycle $\zeta(r)$ defined as follows:
\begin{itemize}
\item Let $\tilde r:X:\to A$ denote the extension of $r$ by setting its value equal to zero on $X_1-Z_1$.
\item The coboundary $d\tilde r :X_2\to A$ restricts to a cocycle on $\bar X_2$ which we denote by $\zeta(r)$.
\end{itemize}

Let us specialize to $Z=Y=NA$.
We take $r$ to be the identity map $\idy:NA\to NA$. Note   that in this case we have
$$
H^1(NA,A) \cong \catGrp(A,A).
$$
Then $\idy:NA\to NA$ can be identified with the identity homomorphism $\idy:A\to A$. We will write
\begin{equation}\label{eq:gamma}
\gamma=\zeta(\idy).
\end{equation}

\Pro{\label{pro:section iff class zero, nonzero implies contextual}
The inclusion $i:NA\to X$ has a section if and only if $[\gamma]=0$. Moreover, if $[\gamma]\neq 0$ then 
  every simplicial distribution $p:X\to {D_R}(NA)$ relative to $NA$ is  contextual.
}
\Proof{
The first statement follows from Lemma \ref{lem:extend iff zero in cohomology}.
For the second statement, assume that $p$ is non-contextual, that is, there exists $d=\sum_r d(r)\delta^r$ with $\Theta(d)=p$. Then there exists $s$ such that $d(s)>0$ and for every $\sigma\in  X_n$ we have
$$
p_\sigma(s_\sigma) = \sum_{r:\,r_\sigma=s_\sigma} d(r) > d(s)>0. 
$$
Taking $\sigma\in (NA)_n$ we have
$$
p_\sigma = (p\circ i)_\sigma = \delta^\sigma
$$
where $\delta^\sigma$ is the delta-distribution peaked at $\sigma$.
Combining the two equations we  obtain
$
\delta^\sigma(s_\sigma)>0
$
which implies that $s_\sigma=\sigma$. This shows that $s:X\to NA$ is a section of $i$. 
}
 
Now, we turn to the equivariant case. 
Let $i:NA\to X$ be a $G$-equivariant inclusion.
We apply the cohomology discussion to the cofiber sequences in Diagram (\ref{dia:cofibrations compare}). For the middle 
cofibration 
sequence
we introduce the following 
cocycle: Using the connecting homomorphism $\zeta:H^1(NA,A)\to H^2(\overline{X\hquo G},A)$ we define
\begin{equation}\label{eq:gamma-G}
\gamma_G=\zeta(\idy).
\end{equation} 
By definition we have 
$$
\bar\iota^*(\gamma_G) =\gamma.
$$

\Pro{\label{pro:section of j vs gamma_G}
The inclusion $j:NA \to X\hquo G$ has a section if and only if $[\gamma_G]=0$. Moreover, if $[\gamma_G]\neq 0$ then every simplicial distribution $p:X\hquo G \to D_R(NA)$ relative to $NA$ is contextual.
}
\Proof{Follows from Proposition \ref{pro:section iff class zero, nonzero implies contextual} applied to $j$.}

\Cor{\label{cor:class nonzero implies equivariant contexuality}
If $[\gamma_G]\neq 0$ then every $G$-equivariant distribution $p:X\to D_R(NA)$ relative to $NA$ is $G$-equivarianly contextual.  
}
\Proof{
By Proposition \ref{pro:section of j vs gamma_G}, $[\gamma_G]\neq 0$ implies that $\tilde p_G$ is contextual. This is equivalent to $p$ being $G$-equivariantly contextual by Proposition \ref{pro:Equivariant contextual iff Borel contextual}.
}  

\Rem{\label{rem:Y subspace NA}
We can slightly generalize one direction of Proposition \ref{pro:section of j vs gamma_G} to the case where $Y$ is a simplicial subset of $NA$: For $Y$ a simplicial subset of $NA$ we have that if {the inclusion} {$i:Y\to X$} has a $G$-equivariant section then $[\gamma_G]=0$. The converse of this statement may not be true since the extension obtained using Lemma \ref{lem:extend iff zero in cohomology} gives a map with target $NA$, not necessarily the simplicial subset $Y$. We are interested in the case where $Y$ is the simplicial subset given by the circle 
$$
S^1_{(a)} \to NA
$$
where the unique non-degenerate $1$-simplex is mapped to a non-trivial element $a\in A$.
}

We have a similar observation for the rightmost cofibration in Diagram (\ref{dia:cofibrations compare}).  Using the connecting homomorphism $\zeta:H^1(NA,A)\to H^2(\bar X\hquo G,A)$ and the projection map $\pi_1:NA\hquo G\to NA$ we define
\begin{equation}\label{eq:gamma-G}
\tilde \gamma_G=\zeta(\pi_1).
\end{equation} 
We have
$$
c^*(\tilde \gamma_G) = \gamma_G.
$$

\Pro{\label{pro:section of i_G vs tilde-gamma_G}
The inclusion $i_G:NA \hquo G \to X\hquo G$ has a section if and only if $[\tilde\gamma_G]=0$. 
Moreover, we have $[\gamma_G]=0$ if and only if $[\tilde \gamma_G]=0$.
}
\Proof{
Proof of the first statement follows from Lemma \ref{lem:extend iff zero in cohomology} applied to $\pi_1:NA\hquo G \to NA$.
The second statement follows from combining Proposition \ref{pro:section of j vs gamma_G} and \ref{pro:section of i_G vs tilde-gamma_G} with Lemma \ref{lem:equivariant section vs normal section}.
}

{Combined with Proposition \ref{pro:section of j vs gamma_G} $[\tilde \gamma_G]$ can also be used to detect contextuality.}

\subsection{Total complex of the Borel construction}

For a simplicial $G$-set $X$, we define a bisimplicial set
\begin{equation}\label{eq:bisimplicial Borel construction}
S_G(X)_{p,q} =  (EG)_p \times_G  X_q. 
\end{equation} 
The diagonal $\diag(S_G(X))$ of this bisimplicial set is the Borel construction $X\hquo G$. In this section we recall the Eilenberg--Zilberg theorem that relates the chains on the diagonal of a bisimplicial abelian group and the total complex of the associated double complex. We apply this result to the   bisimplicial set $S_G(X)$.

Let $L$ be a bisimplicial abelian group. There are two chain complexes associated to this object:
\begin{itemize}
\item The chain complex $C(\diag(L))$ of the diagonal.
\item The total complex $\Tot(C(L))$ of the double complex $C(L)$.
\end{itemize}
By the Eilenberg--Zilber theorem these two chain complexes are chain homotopy equivalent \cite[Chapter 8.5]{weibel1995introduction}. This equivalence is given by the Alexander--Whitney map
$$\Delta: C(\diag(L)) \to \Tot(C(L)) $$ 
and the Eilenberg--Zilber map
$$
\nabla: \Tot(C(L)) \to C(\diag(L))
$$
We are interested in degree $2$ of these maps:
\begin{itemize}
\item For $a\in L_{2,2}$ we have
$$
\Delta_2(a) = ( (d_0^v)^2(a) , d_2^h d_0^v(a) , d_1^hd_2^h(a) ). 
$$
\item For $(a,b,c)\in L_{2,0}\oplus L_{1,1}\oplus L_{0,2}$ we have
$$
\nabla_2(a,b,c) = s_1^vs_0^v(a) + ( s_1^hs_0^v(b) - s_0^h s_1^v(b) ) + s_1^h s_0^h(c).
$$ 
\end{itemize}

We will apply this to the bisimplicial abelian group
$$
L_{p,q} = \ZZ[S_{p,q} ]
$$   
where $S=S_G(X)$ is the bisimplicial set given in Equation (\ref{eq:bisimplicial Borel construction}).
Then the maps $\Delta$ and $\nabla$ in degree $2$ are given by
$$ 
\Delta_2[(1,g_1,g_2),x] = ( [(1,g_1,g_2), (d_0)^2(x)] , [(1,g_1),d_0(x)], [(1),x] )   
$$
and
$$
\begin{aligned}
\nabla_2( [(1,g_1,g_2),x] , [(1,g_1') , x']  ,& [ (1) , x'' ]   ) = \\
 &[(1,g_1,g_2),s_1s_0(x)] +  [(1,g_1',1) , s_0(x') ] -  [(1,1,g_1') , s_1(x') ]  + [(1,1,1) , x'' ] . 
 \end{aligned}
$$ 
{By applying} $\Hom(-,A)$ to the chain complexes {$\Tot(C(L))$ and $C(\diag(L))$} we obtain the dual maps between the associated cochain complexes with coefficients in $A$.
The dual map $\Delta^*: C^*( \Tot(C(L)) ,A) \to C^*(\diag (L),A) $ sends a $2$-cochain $(\alpha,\alpha',\alpha'')$ to 
$$
\theta[(1,g_1,g_2),x] = \alpha[(1,g_1,g_2),(d_0)^2(x)] + \alpha'[(1,g_1),d_0(x)] + \alpha''[(1),x].
$$
On the other hand, the dual map $\nabla^2: C^*(\diag (L),A)\to C^*(\Tot(C(L)),A)$ sends a $2$-cochain $\theta$ to $(\alpha,\alpha',\alpha'')$   where
\begin{equation}\label{eq:nabla2 formula}
\begin{aligned}
\alpha[(1,g_1,g_2),x] &= \theta[(1,g_1,g_2),s_1s_0(x)] \\
\alpha'[(1,g_1'),x'] &= \theta[(1,g_1',1),s_0(x')]- \theta[(1,1,g_1'),s_1(x')]  \\
\alpha''[(1),x''] &= \theta[(1,1,1),x''].
\end{aligned}
\end{equation}



Next, we apply the formulas in Equation (\ref{eq:nabla2 formula}) to the middle and the rightmost cofibrations in Diagram (\ref{dia:cofibrations compare}).
The Borel construction $\bar X
\hquo G$ can be described as the diagonal of the bisimplicial set 
\begin{equation}\label{eq:bar S}
\bar S_{p,q} = S_G(\bar X)_{p,q}= (EG)_p\times_G \bar X_q.
\end{equation}
 Similarly we can describe $\overline{X\hquo G}$ as the diagonal of the following bisimplicial set
\begin{equation}\label{eq:T}
T_{p,q} =  \frac{(EG)_p\times_G X_q }{\ast \times Y_q}.
\end{equation}
The Eilenberg--Zilber map is
natural with respect to the map $T\to \bar S$ of bisimplicial sets.

\Lem{\label{lem:triple representation}
Under the 
Eilenberg--Zilber map
the cocycle $\gamma_G $ is represented by the triple $(0,0,\gamma)$ in the total complex. 
}
\Proof{
Let $\pi_2\colon NA\hquo G \to NA$ 
denote the projection map, recall that $NA\hquo G =NG\times NA$. 
Since the Eilenberg--Zilber map $\nabla$ is natural with respect to maps of bisimplicial sets we can first compute $\nabla^2(\tilde \gamma_G)$ and then pull-back along the corresponding map between the total complexes. Applying the formula in Equation (\ref{eq:nabla2 formula}) to $\bar S$ we obtain that
\begin{equation}\label{eq:computation nabla2}
\nabla^2(\tilde \gamma_G) = (0,0,\gamma).
\end{equation}
 Pulling back this cochain along $T\to \bar S$ gives the desired result. 
To  obtain 
Equation (\ref{eq:computation nabla2}) we first consider how the cocycle $\tilde \gamma_G$ is defined. In the cohomology long exact sequence   we begin with $\pi_1$ regarded as a $1$-cochain $\pi_1:A\times G\to A$, again given by projection, and define $\tilde\pi: X_1\times_G G^2 \to A$ by lifting $\pi_1$ on the rest of the $1$-simplices by setting it equal to zero. Then applying the coboundary in $ X\hquo G$ we obtain $\tilde \gamma_G$. Looking at the description of $d(\tilde\pi_1)$ and the formula in Equation (\ref{eq:nabla2 formula}) we obtain Equation (\ref{eq:computation nabla2}).
In more details, the map $\nabla^2(\tilde \gamma_G)$ consists of three factors. For  $x\in X_0$, $x'\in X_1$ and $x''\in X_2$, we have
\begin{enumerate}
\item $\tilde \gamma_G((1,g_1,g_2),s_1s_0(x))=0$ since $\tilde\pi$ restricted to a degenerate simplex of the form $s_1s_0(x)$ is zero.
\item $\tilde \gamma_G((1,g_1',1),s_0(x'))-\tilde \gamma_G((1,1,g_1'),s_1(x'))=0$: We have 
$$
\begin{aligned}
&\tilde \gamma_G((1,g_1',1),s_0x')-\tilde \gamma_G((1,1,g_1'),s_1x')\\
&= \tilde{\pi}((g_1',1),x')-\tilde{\pi}((1,g_1'),x') +\tilde{\pi}((1,g_1'),s_0d_1x')
 -\tilde{\pi}((1,g_1'),s_0d_0x')+\tilde{\pi}((1,g_1'),x')-\tilde{\pi}((1,1),x')\\
&=\tilde{\pi}((g_1',1),x')+\tilde{\pi}((1,g_1'),s_0d_1x')-\tilde{\pi}((1,g_1'),s_0d_0x')-\tilde{\pi}((1,1),x').
\end{aligned}
$$
In the sum above, the second and third component are $0$, as the second factor in $\tilde{\pi}$ is a degeneracy of the vertex of $X$. Thus we are left with the sum
\[
\tilde{\pi}((g_1',1),x')-\tilde{\pi}((1,1),x').
\]
However, this sum also becomes zero since if $x'\in (NA)_1$, then $G$ acts trivially on $x'$ and we have that $((g_1',1),x')=((1,1),x')$ in $X\hquo G$. Otherwise, if $x'\not\in (NA)_1$, then also $g_1'\cdot x\not\in (NA)_1$, giving that both components of the sum are zero.

\item $\tilde \gamma_G((1,1,1),x'')=\gamma(x'')$ since $\tilde{\gamma}_G$ pulls-back to $\gamma$ along $c\circ \bar\iota$ in Diagram \eqref{dia:cofibrations compare}.
\end{enumerate}

}

\Def{\label{def:phi class}
Assume that 
$[\gamma]=0$ in $H^2(\bar X,A)$.
Consider a $1$-cochain $s:\bar X_1 \to A$ such that $d^vs=\gamma$ {in the total complex of $\bar S$ (see Equation (\ref{eq:bar S}))}. We introduce a $(1,1)$-cochain $\tilde \phi:G\times \bar X_1 \to A$  
in  $\Tot(C(\ZZ[\bar S]))$
by
$$
\tilde \phi = - d^hs,
$$ 
where we identify $G^2\times_G\bar X_1$ with $G\times \bar X_1$.
Then we  define 
$$
\phi = c^*(\tilde\phi),
$$
a $(1,1)$-cochain 
{in the total complex of $T$ (see Equation (\ref{eq:T})).}
}


{Note that one important case where $\phi$ can be defined is when $[\gamma_G]=0$. Since $\gamma_G$ pulls-back to $[\gamma]$ under $\bar\iota$ we have $[\gamma]=0$ in this case and thus we can define $\phi$.
}

{
\begin{thm}\label{thm:gammaG zero then beta and phi trivial}
Let $X$ be a simplicial $G$-set  and $i:NA\to X$ be an inclusion of simplicial $G$-sets where the action on $NA$ is trivial. 
Assume that the class
$[\gamma_G]$ 
defined in Equation (\ref{eq:gamma-G}) is zero, so that $\phi$ is defined (see Definition \ref{def:phi class}).
Then there exists $r:\bar X_1\to A$ such that $d^vr=0$ and $d^hr=\phi$. 
\end{thm}
\begin{proof}
By Lemma \ref{lem:triple representation} $[\gamma_G]=0$ if and only if $[(0,0,\gamma)]$ is zero in the total complex. Since $[\beta]=0$ we can trivialize it and define $\phi$. That is, we have $[(0,0,\gamma)]=[(0,\phi,0)]$ in the total complex. Then this class is zero if and only if there exists $r:\bar X_1\to A$ with the desired property.
\end{proof}

\Rem{\label{rem:equiv contextuality double complex}
{\rm
By the proof of Theorem \ref{thm:gammaG zero then beta and phi trivial} we can show that $[\gamma_G]\neq 0$ by working in the total complex. 
We have $[\gamma_G]=0$ if and only if  
\begin{enumerate}
\item there exists $s:\bar X_1\to A$ such that $d^vs=\gamma$, and
\item there exists $r:\bar X_1\to A$ such that $d^vr=0$ and $d^hr = \phi$ (where $\phi=-d^hs$).
\end{enumerate}
Note that {$\tilde s=(0,r+s)$} satisfies $d\tilde s=(0,0,\gamma)$. 
We can apply this observation to the case when $[\gamma]=0$. Then $s$ as specified in (1) exists. Thus $[\gamma_G]=0$ if and only if $2$ as specified in (2) exists. Combined with Proposition \ref{pro:section of j vs gamma_G} we can use this observation to detect contextuality.
}
}

}


\subsection{{Example:} Torus relative to the diagonal}
\label{sec:torus-rel-to-diag} 

We revisit the example $X=S^1\times S^1$ and $Y=S^1$ from Section \ref{sec:torus-with-involution}.
Let $i:Y\to X$ be the diagonal map {sending the non-degenerate $1$-simplex of $Y$ to $x$ in Figure (\ref{fig:simp dist torus})}.
A simplicial distribution $p:X\to D(Y)$ is relative to $Y$ if and only if $p_x=\delta^1$. 
{Recall that  simplicial distributions on $(X,Y)$ are given by $(t_1,t_2)\in [0,1]^2$ such that $1-(t_1+t_2)\geq 0$. The relativity condition then implies that $t_1+t_2=1$.}
In the $G$-equivariant case, i.e., when $t_1=t_2\in [0,1/2]$, this forces $t_1=1/2$. Therefore there exists a unique $G$-equivariant simplicial distribution relative to $Y$ in this case.

When we take the Borel construction and then restrict to the $2$-skeleton {(this gives $W=(X\hquo G)^{(2)}$)} the inclusion map $j:Y\to W$ sends the unique non-degenerate $1$-simplex of $Y$ to the edge $x(0)$ (the diagonal edge in the 
space in Figure (\ref{fig:borel1})). 
Again there is a unique simplicial distribution relative to $Y$ given by $t_1=1/2$ {by the isomorphism in (\ref{dia:iso for torus}).}  

Now, we consider the cofiber sequence
\begin{equation}\label{eq:cofiber sequence}
Y \xrightarrow{j} W \to \overline W
\end{equation}
Let $[ \gamma_G]$ denote the cohomology class obtained as the image of $[s_1]\in H^1(Y,\ZZ_2)$ under the connecting homomorphism, where $s_1:Y\to N\ZZ_2$ sends the unique non-degenerate $1$-simplex to the edge labeled by $1$. 
We can compute the cochain $\gamma_G$ using the construction of the connecting homomorphism.
More explicitly, we have $\gamma_G(\sigma)=1$ for 
$$
\sigma \in \set{\sigma_0(00), \sigma_1(00), \sigma_0(11), \sigma_1(11) },
$$ 
and otherwise zero. We see that $\gamma_G$ is non-zero on the upper triangle of the 
surface {in Figure (\ref{fig:borel3})},  which is a closed surface in $\overline W$, therefore $[\gamma_G]\neq 0$. Then Proposition \ref{pro:section iff class zero, nonzero implies contextual} implies that the simplicial distribution with $t_1=1/2$ is contextual.
  
Using Lemma \ref{lem:Borel preserve relative} we conclude that every $G$-equivariant simplicial distribution $p:X\to D(Y)$ relative to $Y$ is $G$-equivariantly contextual.

Next we want to understand the representation of $\gamma_G$ in the total complex. 
Note that here we regard $\gamma_G$  as a cocycle on $\overline{(X\hquo G)}$, the cofiber of $Y\to X\hquo G$,  rather than the subspace $\overline W$. 
Since $[\gamma]=0$ there exists $s:\bar X_1\to \ZZ_2$ such that $d^vs=\gamma$. For  concreteness we will choose
\begin{equation}\label{eq:s}
s(x_0)=0\;\;\text{ and }\;\; s(x_1)=1.
\end{equation}
In the total complex we have
$$
(0,0,\gamma)-d(0,s) =  (0,0-d^hs,\gamma-d^vs) = (0,-d^hs,0).
$$
The cocycle $\phi=-d^hs$, which can be regarded as a function $\ZZ_2\times \bar X_1\to \ZZ_2$, is given by
\begin{equation}\label{eq:phi-torus}
\begin{aligned}
\phi(g,x_i) = s(g\cdot x_i) +s(x_i) = \left\lbrace
\begin{array}{cc}
1 & g=1 \\
0 & \text{otherwise.} 
\end{array}
\right.
\end{aligned}
\end{equation}

\section{Quantum distributions}
\label{sec:quantum distributions} 
 
{Equivariant simplicial distributions arise naturally in quantum theory. In this section we provide an approach based on partial groups, which is complementary to the cofibration approach of the previous sections. Then we apply our theory to equivariant simplicial distributions that are obtained from quantum theory. {In Section \ref{sec:Mermin star} we consider an important example, {known as the Mermin star construction}, 
{which also has applications to}
 measurement-based quantum computing  \cite{raussendorf2016cohomological}.} {We show that this construction produces equivariantly contextual simplicial distributions and this kind of contextuality can be detected by the cohomology classes we construct in the theory. }
}

\subsection{Extensions of partial groups}

We recall the theory of extensions of partial groups from \cite{broto2021extension}.
There are two important groups associated to a partial group:
\begin{itemize}
\item The normalizer of the partial group $M$ is defined by
$$
N(M) = \set{ \theta\in M_1:\,  \exists\, \idy_M \xrightarrow{\theta} f_\theta   }.
$$
\item The center of $M$ is the subgroup $Z(M)\subset N(M)$ consisting of those $\theta\in M_1$ such that $\idy_M \xrightarrow{\theta} \idy_M$.
\end{itemize}
We are interested in the case where $M=NA$. We have $Z(NA)=N(NA)=A$.
Extensions of partial groups are described by fiber bundles.
Let $N$ be another partial group. The extension of $M$ by $N$ is defined to be a fiber bundle
$$
N\to E\to M
$$ 
In the theory of partial groups the total space $E$ can be described by a twisted product. 
An $M$-twisting pair for $N$ is a pair of functions $(t,\beta)$ where
$$
t:M_1 \to \Aut(N)\;\;\;\;\text{ and }\;\;\;\; \beta:M_2\to N(N)
$$
satisfying
\begin{itemize}
\item $\beta(x,y)$ determines a homotopy $t(x\cdot y) \to t(x)\cdot t(y)$, 
\item $t(1)=\idy$ and $\beta(x,1)=1=\beta(1,x)$ for all $x\in M_1$,
\item for all $(x,y,z)\in M_3$,
$$
t(x)(\beta(y,z)) \cdot \beta(x,y\cdot z) = \beta(x,y) \cdot \beta(x\cdot y,z).
$$ 
\end{itemize} 
The pair $(t,\beta)$ can be used to define a twisting function $\tau$ and the total space can be described as the twisted product $E=N\times_\tau M$; see \cite[Theorem 4.5]{broto2021extension}. We will not go into the description of $\tau$ in general. We will be more explicit when $N=NA$, the case of interest for us. In this paper we will need  two extreme cases of extensions:
\begin{itemize}
\item $\beta$ is trivial: In this case the twisted product is simply denoted by $N\ltimes M$. Our main example will come from an action of a group $G$ on $M$. We will consider $NG\ltimes M$.

\item $t$ is trivial: Such extensions are called central. We will consider the case where $N=NA$. Note that in this case $\beta:M_2\to A$ is precisely a $2$-cocycle. We will write $NA\times_\beta M$ for the twisted product.
\end{itemize}

Central extensions of partial groups generalize central extensions of groups. 
Let $K$ be a group $A\subset K$ be a central subgroup. We have a central group extension
\begin{equation}\label{eq:central extension}
{0} \to A \to K \xrightarrow{\epsilon} \bar K \to 1
\end{equation}
Consider a set-theoretic section 
$$\eta:\bar K\to K$$ 
of $\epsilon$ satisfying $\eta(1_{\bar K})=1_K$.
Then we can define a cocyle
$$
\beta(\bar k_1,\bar k_2) = \eta(\bar k_1) \eta(\bar k_2) \eta(\bar k_1 \bar k_2)^{-1}.
$$
The associated cohomology class $[\beta]$ classifies the group extension. 
We can also regard this extension as an extension of partial groups. 
Let $A\times_\beta \bar K$ denote the group whose elements consists of pairs $(a,\bar k)$ together with the multiplication rule twisted by $\beta$:
\begin{equation}\label{eq:product beta twisted}
(a_1,\bar k_1)\cdot (a_2,\bar k_2) = (a_1+a_2+\beta(\bar k_1,\bar k_2),\bar k_1 \bar k_2).
\end{equation}
Then there is a commutative diagram
\begin{equation}\label{dia:group extension as partial group}
\begin{tikzcd}[column sep=huge, row sep=large]
NA \arrow[r,equal] \arrow[d,"i"] & NA \arrow[d] \\
NA\times_\beta N\bar K \arrow[d,"\pi"] \arrow[r,"\cong"] & N(A\times_\beta \bar K) \arrow[d]\\
N\bar K \arrow[r,equal] & N\bar K
\end{tikzcd}
\end{equation}
The section $\eta$ also defines a pseudo-section of the left-hand partial group extension, which by a slight abuse of notation, we still denote by $\eta:N\bar K\to N\ZZ_d\times_\beta N\bar K$. {(Pseudo-section means that the map is compatible with the simplicial structure maps except $d_0$.)} 
{The isomorphism is given by 
$$
((a_1,\cdots,a_n),(\bar k_1,\cdots,\bar k_n)) \mapsto ((a_1,\bar k_1),\cdots,(a_n,\bar k_n))
$$}

{For an arbitrary central partial group extension we will allow the base space to be a partial group: $NA\to NA \times_\beta M \to M$ where $M$ is a partial group. Equation (\ref{eq:product beta twisted}) also holds for the $1$-simplices of the total space as the $\Pi_2$ map coincides with the $d_1$-face of the twisted product.}
 
\Lem{\label{lem:beta zero iff pi splits iff i splits}
For a central partial group extension
$$
NA \xrightarrow{i} NA \times_\beta M \xrightarrow{\pi} M 
$$
the following are equivalent:
\begin{enumerate}
\item $[\beta]=0$ in $H^2(M,A)$.
\item The map $i$ splits.
\item The map $\pi$ splits. 
\end{enumerate} 
}
\Proof{
The class $[\beta]$ is zero if and only if 
the twisted product is isomorphic to the trivial one $E\cong NA\times M${; see for example \cite[App. B]{chung2023simplicial}}.
Both (2) and (3) are equivalent to this latter condition.
}

\subsection{Group action on partial groups}
\label{sec:Group action on partial groups}

Let $G$ be a group. Consider the central partial group extension
$$
NA \to E \xrightarrow{\pi} M
$$ 
where $E=NA\times_\beta M$.
Our group will act on the total space in such a way that its restriction to the fiber is trivial and moreover the action is compatible with the $NA$-action on the total space. Note that this fiber bundle is in fact a principal $NA$-bundle: the fiber is not just a partial group, it is a simplicial abelian group. More precisely, the action is given by  a homomorphism $\varphi:G\to \Aut(NA\times_\beta M)$, which we simply denote by $g\cdot x = \varphi(g)(x)$, such that 
\begin{itemize}
\item $g\cdot a = a$ for all $a\in (NA)_n$, and
\item $g\cdot(a\cdot x) = a\cdot (g\cdot x)$. 
\end{itemize}
{Thus $G$ acts on $E$ by partial group automorphisms fixing $NA$.}
Here $NA$ acts freely on the total space by left multiplication on the first coordinate of the twisted product.
Then we have
\begin{equation}\label{eq:action g and Phi}
g\cdot (a,m) = (a+\Phi_g(g\cdot m),g\cdot m).
\end{equation}
for some function $\Phi_g:M_1\to A$.  We will think of this as a function $\Phi:G\times M_1\to A$. Our approach will be to regard this as a cochain in a double complex.

First we begin with the partial group extension associated to the group action:
\begin{equation}\label{eq:Group action extension of partial monoids}
NA \to NG\ltimes E \xrightarrow{\tilde \pi} NG\ltimes M
\end{equation}
where $\tilde{\pi}\colon NG\ltimes E\to NG\ltimes M$ is defined on $1$-simplices by 
$\tilde{\pi}((g,x))=(g,\pi(x))$.
The associated twisting pair $(t,\beta)$ is given by $t=\varphi$ and $\beta=0$.
Our next goal is to identify the partial group extension associated to the group action with the Borel construction.

\begin{lem}
\label{lem:iso-Borel-and-semidirect-product}
There is an isomorphism of simplicial sets
\[
M \hquo G\cong NG\ltimes M
\]
given by  $S: M\hquo G \to NG\ltimes M$ defined by
\begin{align*} 
[(g_0,g_1,\ldots,g_n),(m_1,\cdots,m_n)]&\mapsto ((g_0^{-1}m_1,g_1),(g_1^{-1}g_0^{-1}m_2,g_2),\cdots,((g_{n-1}g_{n-2}\cdots g_2g_1)^{-1}g_0^{-1}m_n,g_n)]
\end{align*}
and $T: NG\ltimes M \to M\hquo G$ defined by
\begin{align*} 
((m_1,g_1),\cdots,(m_n,g_n))&\mapsto [(1,g_1,\cdots,g_n),(m_1,g_1m_2,\cdots,g_{n-1}\cdots g_2g_1m_n)].
\end{align*}
\end{lem}
\begin{proof}
Follows from direct verification.

\end{proof}

We note that the isomorphism given in the lemma above can be extended to   extensions of partial groups.

\begin{pro}\label{pro:Borel construction as a semidirect product}
There is a commutative diagram of partial group extensions
$$
\begin{tikzcd}[column sep=huge,row sep =large]
NA \arrow[r] \arrow[d,equal] & E\hquo G \arrow[d,"\cong"] \arrow[r,"p"] & M\hquo G \arrow[d,"\cong"] \\
NA \arrow[r] & NG\ltimes E \arrow[r, "\tilde{\pi}"] & NG\ltimes M
\end{tikzcd}
$$

\end{pro}
\Proof{
We need only to prove that the right-hand square indeed commutes. Note that the vertical isomorphisms are given by $S$ in Lemma \ref{lem:iso-Borel-and-semidirect-product}.
Let $[(1,g_1,\cdots,g_n),(x_1,\cdots,x_n)]\in (E\hquo G)_n$. Then we have that
\begin{align*}
S\circ p([(1,g_1,\cdots,g_n),(x_1,\cdots,x_n)])&= S((1,g_1,\ldots,g_n),\pi(x))\\
&=((g_1,\pi(x_1)),(g_2,g_1^{-1}\pi(x_2)),\cdots,(g_n, (g_{n-1}g_{n-2}\cdots g_{1})^{-1}\pi(x))).
\end{align*}
On the other hand,
\begin{align*}
\tilde{\pi}\circ S([(1,g_1,\cdots,g_n),(x_1,\cdots,x_n)])&=\tilde{\pi}((g_1,x_1),(g_2,g_1^{-1})x_2,\cdots,(g_n, (g_{n-1}g_{n-2}\cdots g_{1})^{-1}x))\\
&=((g_1,\pi(x_1)),(g_2,g_1^{-1}\pi(x_2)),\cdots,(g_n, ((g_{n-1}g_{n-2}\cdots g_{1})^{-1}\pi(x))).
\end{align*}
Thus the claim follows.
}

Recall that $M\hquo G$ is the diagonal of the bisimplicial set $S_G(M)_{p,q}=(EG)_p\times_G M_q$ {defined in Equation (\ref{eq:bisimplicial Borel construction})}. We regard $\Phi$ as a $(1,1)$-cochain in the double complex $C^*( \ZZ[S_G(M)] ,A)$. Similarly we can regard $\beta:M_2\to A$ as a $(0,2)$-cochain in the same complex.  
 
\Lem{
The function $\Phi$ satisfies: 
\begin{enumerate}
\item $d^v\Phi = d^h\beta$,
 
\item $d^h\Phi=0$.
\end{enumerate}
} 
\begin{proof}
Automorphisms of $E=NA\times_\beta M$ are determined by their restriction to $1$-simplices (part (1), Lemma \ref{lem:maps into partial monoids and homotopy}). 
We have a commutative diagram
$$
\begin{tikzcd}[column sep=huge,row sep =large]
(NA\times_\beta M)_2 \arrow[d,"g\cdot"] \arrow[r,"d_2\times d_0"] & (NA\times_\beta M)_1 \times (NA\times_\beta M)_1 \arrow[d,"g\cdot "] \\
(NA\times_\beta M)_2 \arrow[r,"d_2\times d_0"] & (NA\times_\beta M)_1 \times (NA\times_\beta M)_1
\end{tikzcd}
$$

By commutativity of the diagram we obtain
$$
g\cdot ((0,0),\sigma) = ( (\Phi_g(g\cdot d_2\sigma), \Phi_g(g\cdot d_0\sigma)+\beta(\sigma)-\beta(g\cdot \sigma) ), g\cdot \sigma  ).
$$
Now, $g\cdot$ is a simplicial set map and since the twisted product is determined by a $2$-cocycle it suffices to require compatibility with the face maps $d_i$ from $2$-simplices.  Compatibility with $d_1$ requires that
$$ 
d_1(g\cdot ((0,0),\sigma)) =   ( \Phi_g(g\cdot d_2\sigma)+ \Phi_g(g\cdot d_0\sigma)+\beta(\sigma)-\beta(g\cdot \sigma) , d_1(g\cdot \sigma)  )
$$
is equal to
$$
g\cdot(d_1 ((0,0),\sigma) ) = (\Phi_g(g\cdot d_1\sigma), g\cdot (d_1\sigma) )
$$
which gives $d^v\Phi = d^h\beta$. 
A similar computation shows that $g\cdot $ is compatible with $d_0$ and $d_2$. Finally $d^h\Phi=0$ follows from the requirement that $g\cdot (h\cdot ((0,0),\sigma)) = gh\cdot ((0,0),\sigma)$.
Indeed, using Equation (\ref{eq:action g and Phi}) we have that
\[
\begin{aligned}
g\cdot (h\cdot ((0,0),\sigma)) &=g\cdot (\Phi_h(h\cdot \sigma),h\cdot \sigma)=(\Phi_g(gh\cdot \sigma)+\Phi_h(h\cdot \sigma), gh\cdot \sigma) \\
&= (\Phi_{gh}(gh\cdot \sigma), gh\cdot \sigma)\\
&= gh\cdot ((0,0),\sigma).
\end{aligned}
\]
This gives the identity
\[
\Phi_g(gh\cdot \sigma)+\Phi_h(h\cdot \sigma)=\Phi_{gh}(gh\cdot \sigma).
\]
from which the result follows: 
$
d^h\Phi(g,h,\sigma)=\Phi_g(gh\cdot \sigma)+\Phi_h(h\cdot \sigma)-\Phi_{gh}(gh\cdot \sigma)=0
$.
\end{proof}

Let $\eta:M\to NA\times_\beta M$ denote the pseudo-section defined by $\eta(x)=(0,x)$.  
Then 
we have 
\begin{equation}\label{eq:cocycle from eta}
\begin{aligned}
\eta(x) \cdot \eta(y) \cdot \eta(x\cdot y)^{-1}   &=  (0,x)\cdot (0,y) \cdot (0,x\cdot y)^{-1} \\
&=  (\beta(x,y),x\cdot y) \cdot (-\beta(x\cdot y,x^{-1}\cdot y^{-1}),y^{-1}\cdot x^{-1})  \\
&=  (\beta(x,y) ,1)  
\end{aligned}
\end{equation}
and
\begin{equation}\label{eq:phi from eta}
\begin{aligned}
(g\cdot \eta(g^{-1}\cdot x))\cdot \eta(x)^{-1}  & = (g\cdot (0,g^{-1}\cdot x)) \cdot (0,x)^{-1} \\
&=  (\Phi_g(x),x) \cdot (-\beta(x,x^{-1}),x^{-1}) \\
&= (\Phi_g(x),1).
\end{aligned}
\end{equation}

\Lem{\label{lem:Phi zero iff equivariant section}
There exists $s:M_1\to A$ such that $d^hs=\Phi$ if and only if $\pi$ admits a $G$-equivariant pseudo-section.
}
\Proof{ 
We firstly note that the cohomology class of $\Phi$ does not depend on the choice of the pseudo-section $\eta$. Indeed, let $\eta_1$ and $\eta_2$ be two pseudo-sections and let $\Phi^{(1)}$ and $\Phi^{(2)}$ be the corresponding functions defined using Equation \eqref{eq:phi from eta}. Since for every $x\in M_1$ we have that $\pi(\eta_1(x))=\pi(\eta_2(x))$, there exists a function $\alpha\colon M_1\to A$ such that $\eta_2(x)=\alpha(x)\eta_1(x)$. Recall that we assume that the action of $G$ on $A$ is trivial. Therefore we have that for any $g\in G$, $x\in M_1$
\begin{align*}
\Phi^{(2)}(g,x)&=g\cdot\eta_2(g^{-1}x)\eta_2(x)^{-1}\\
&= \alpha(g^{-1}x)g\cdot\eta_1(g^{-1}x)\alpha(x)^{-1}\eta_1(x)^{-1}\\
&=\alpha(g^{-1}x)-\alpha(x)+{\Phi^{(1)}(g,x)}.
\end{align*}
We obtain that the difference of the cocycles $\Phi^{(1)}$ and $\Phi^{(2)}$ is a coboundary of the cochain $\alpha$. So the class of $\Phi$ does not depend on the choice of the pseudo-section.

From Equation \eqref{eq:phi from eta} it follows directly that if $\eta$ {is} $G$-equivariant, then the cocycle $\Phi$ become zero. By the observation above we see that if there is any equivariant pseudo-section, then the cocycle $\Phi$ differs from zero by a coboundary, thus its cohomology class is zero.

Assume now that $[\Phi]=0$, that is, there exists a function $a\colon M_1\to A$ such that  for every $g\in G$ and $x\in M_1$ we have that
\[
a(g^{-1}x)-a(x)=g\cdot\eta(g^{-1}x)\eta(x)^{-1}.
\]
Now define the function $\lambda\colon M_1\to A$ by $\lambda(x)=\eta(x)a(x)^{-1}$. From the equation above and the fact that $\pi$ maps $A$ to the trivial partial subgroup of $M$ it follows that $\lambda$ is an equivariant pseudo-section. Thus the claim follows.
}

{T}he pseudo-section $\eta_G:M\hquo G \to E\hquo G$ defined by 
$$\eta_G([(1,g_1,\cdots,g_n),x)]) =[(1,g_1,\cdots,g_n),\eta(x)]$$
gives a $2$-cocycle $\beta_G: (M\hquo G)_2\to A$ defined by a formula similar to Equation (\ref{eq:cocycle from eta})
\begin{equation}\label{eq:cocycle from eta tilde}
[(1,1),\beta_G(\sigma,x,y)] =  \eta_G(x) \cdot \eta_G(y) \cdot \eta_G(x\cdot y)^{-1},
\end{equation}
where $[\sigma,(x,y)]\in (EG)_2\times_G M_2$, such that the principal $NA$-bundle in  (\ref{eq:Group action extension of partial monoids}) is classified by $[\beta_G]$.

\Rem{\label{rem:isomorphism to semidirect product}
{\rm
Using Lemma \ref{lem:iso-Borel-and-semidirect-product} we can identify  $M\hquo G \cong NG\ltimes M$.
Thus for $\sigma=(1,g_1,g_2)\in (EG)_2$  the element $[(\sigma,(x_1,x_2))]$ is represented by $((g_1,x_1),(g_2,g_1^{-1}x_2))\in (NG\ltimes E)_2$. 
The function $\eta_G\colon NG\ltimes M\to NG\ltimes E$ is then given by

\[
\eta_G((g_1,x_1),\cdots,(g_n,x_n))=((g_1,\eta(x_1)),\cdots,(g_n,\eta(x_n))).
\]
Using this we obtain
\begin{equation}\label{eq:beta-G}
\begin{aligned}
{(1,\beta_G((g_1,x_1),(g_2,x_2))} 
&=\eta_G(g_1,x_1)\cdot\eta_G(g_2,x_2)\cdot\eta_G(g_1g_2, x_1\cdot(g_1x_2))^{-1}\\
&= (g_1,\eta(x_1))\cdot (g_2,\eta(x_2))\cdot (g_1g_2, \eta(x_1\cdot(g_1x_2)))^{-1}\\
&= (g_1g_2,\eta(x_1)\cdot(g_1\eta(x_2)))\cdot((g_1g_2)^{-1},(g_1g_2)^{-1}\cdot \eta(x_1\cdot(g_1x_2))^{-1})\\
&= (1, \eta(x_1)\cdot(g_1\eta(x_2))\cdot\eta(x_1\cdot(g_1x_2))^{-1}).
\end{aligned}
\end{equation}

}
}

\Lem{\label{lem:description of class of the Borel construction}
The cocycle $\beta_G$ is represented by the triple $(0,\Phi,\beta)$ in the total complex.
}
\Proof{
We will use the identification in Remark \ref{rem:isomorphism to semidirect product} {and Equation (\ref{eq:beta-G})}.
The claim follows from the formulas given in Equation \eqref{eq:nabla2 formula}.
For the first factor, recall that $M$ is a reduced simplicial set (see Definition \ref{def:partial monoid}) and for the unique vertex $\ast\in M_0$ we have $s_0(\ast)=1\in M_1$. Thus for $\sigma=(1,g_1,g_2)\in EG_2$ we compute that
\begin{align*}
\alpha(\sigma,\ast))&=\beta_G(\sigma,s_1s_0(\ast))\\
&=\beta_G((g_1,1),(g_2,1))\\
&= (1,\eta(1)\cdot(g_1\eta(1))\cdot\eta(1\cdot (g_1\cdot 1))^{-1})\\
&= (1,0).
\end{align*} 
We obtain that $\alpha(\sigma,\ast)=0$.
The second factor comes from the following computation for $g\in G$ and $x\in M_1$:
\begin{align*}
\alpha'((1,g),x)&=\beta_G((1,g,1),s_0 x)-\beta_G((1,1,g),s_1x)\\
&=\beta_G((g,1),(1,g^{-1}x)))-\beta_G((1,x),(g,1))\\
&=(1,\eta(1)\cdot(g_1\eta(g^{-1}x))\cdot\eta(1)^{-1})-(1,\eta(x)\cdot \eta(1)\cdot \eta(x)^{-1})\\
&= (1,\eta(1)\cdot(g_1\eta(g^{-1}x))\cdot\eta(1)^{-1})\\
&=(1,\Phi_g(x)).
\end{align*}
Therefore $\alpha'((1,g),x)=\Phi_g(x)$. Finally the third factor can be computed as
\begin{align*}
\alpha''[(1),(x_1,x_2)] &= \beta_G((1,1,1),(x_1,x_2)) \\
&= \beta_G((1,x_1),(1,x_2)) \\ 
&= (1,\eta(x_1)\cdot\eta(x_2)\cdot\eta(x_1x_2)^{-1})\\
&= (1,\beta(x_1,x_2)).
\end{align*}
}


\Thm{\label{thm:betaG zero iff there exists s iff equivariant section}
{Let $NA \xrightarrow{x} E \xrightarrow{\pi} M$ be a central partial group extension and $G$ be a group acting on $E$ partial group {automorphisms} that fix $NA$.}
The following are equivalent:
\begin{enumerate}
\item The class $[\beta_G]$ is zero. 
\item There exists $s:M_1\to A$ such that $d^vs=\beta$ and $d^hs=\Phi$.
\item The map $\pi$ admits a $G$-equivariant section.
\end{enumerate}
}
\Proof{
We begin by showing the equivalence of (1) and (2).
By Lemma \ref{lem:description of class of the Borel construction} $[\beta_G]=0$ if and only if $[(0,\Phi,\beta)]=0$ in the total complex. This condition is equivalent to the existence of $s:M_1\to A$ and $r:G\to A$ such that $d(r,s)=(0,\Phi,\beta)$. Unraveling the coboundary in the total complex gives
$$
d(r,s) = (d^hr, d^hs,d^vs)
$$
where we used $d^vr=0$. Thus $[(0,\Phi,\beta)]=0$ if and only if $d^hr=0$, $d^hs=\Phi$ and $d^vs=\beta$. We obtain the desired result since we can take $r=0$.

Equivalence of (2) and (3) follows from Lemma \ref{lem:beta zero iff pi splits iff i splits} and Lemma \ref{lem:Phi zero iff equivariant section}.
}


{
\Rem{
{\rm
In practice we will not carry over the trivial part in the Equations (\ref{eq:cocycle from eta}-\ref{eq:beta-G}) and simply write
$$
\begin{aligned}
\beta(x,y) & = \eta(x)\cdot \eta(y) \cdot \eta(x\cdot y)^{-1} \\
\Phi_g(x) & = (g\cdot \eta(g^{-1}\cdot x)) \cdot \eta(x)^{-1} \\
 \beta_G(x,y) & = \eta_G(x)\cdot \eta_G(y) \cdot \eta_G(x\cdot y)^{-1} .
\end{aligned}
$$

}
}
}

\subsection{Comparison to cofibration}

There is a canonical comparison diagram
\begin{equation}\label{dia:comparison cofib fib}
\begin{tikzcd}[column sep=huge,row sep =large]
NA \arrow[r,equal] \arrow[d,"i"] & NA \arrow[d,"i"] \\
E \arrow[r,equal] \arrow[d,"q"] & E \arrow[d,"\pi"] \\
\bar E \arrow[r,"\alpha"] &  M 
\end{tikzcd}
\end{equation}
where the left sequence is a cofiber sequence and right sequence is a central extension of partial groups.

\Lem{\label{lem:beta pullbacks to gamma}
We have $\alpha^*([\beta])=[\gamma]$.
}
\Proof{
Consider the Serre spectral sequence for the right vertical sequence {in Diagram (\ref{dia:cofibrations compare})} with coefficients in $A$.
Let $\id$ be the identity homomorpshim of $A$. Then we have that $[\id]\in H^1(A,A)$ and that $d_2([\id])=[\beta]${; see \cite{adem2013cohomology}}. This differential is the transgression homomorphism and it is defined by the following diagram {(see \cite[Sec. 6.2]{mccleary2001user})}:
\[
\begin{tikzcd}[ampersand replacement = \&,column sep=huge,row sep =large]
0\rar\dar\& H^2(M,A)\rar\dar["\alpha^\ast"]\& H^2(M,A)\rar\dar["\pi^\ast"]\& 0\\
H^1(A,A)\rar["\zeta"]\& H^2(\bar{E},A)\rar\& H^2(E,A)\&
\end{tikzcd}
\]
This diagram comes from the map of cofibrations:
\[
\begin{tikzcd}[ampersand replacement=\&,column sep=huge,row sep =large]
NA \arrow[r] \arrow[d,"i"] \& \ast \arrow[d] \\
E \arrow[r,,"\pi"] \arrow[d] \& M \arrow[d,equal] \\
\bar E \arrow[r,"\alpha"] \&  M 
\end{tikzcd}
\]
By the definition of the transgression we obtain that $\alpha^\ast([\beta])=\zeta([\id])=[\gamma]$.
}


\Pro{\label{pro:equivalence section, zero cofib, zero fib}
The following are equivalent.
\begin{enumerate}
\item $i:NA\to E$ has a section.
\item $[\gamma]=0$ in $H^2(\bar E,A)$.
\item $[\beta]=0$ in $H^2(M,A)$.
\end{enumerate}
}
\Proof{
Equivalence of (1) and (2) follows from Proposition \ref{pro:section iff class zero, nonzero implies contextual} and the equivalence of (1) and (3) follows from Lemma \ref{lem:beta zero iff pi splits iff i splits}.
}

\Cor{\label{cor:beta nonzero then contextual}
If $[\beta]\neq 0$ then every simplicial distribution $p:E\to D_R(NA)$ relative to $NA$ is contextual.}
\Proof{
Follows from Proposition \ref{pro:section iff class zero, nonzero implies contextual} and Proposition \ref{pro:equivalence section, zero cofib, zero fib}.
}

We can also apply the discussion of $G$-equivariant contextuality. In this case the comparison diagram is given as follows
\begin{equation}\label{dia:comparison cofib1 cofib 2 fib}
\begin{tikzcd}[column sep=huge,row sep =large]
NA \arrow[r] \arrow[d] & NA\hquo G\arrow[d] \arrow[r] & NA \arrow[d] \\
E\hquo G \arrow[r,equal] \arrow[d] & E\hquo G \arrow[r,equal] \arrow[d] & E\hquo G \arrow[d] \\
\overline{E\hquo G} \arrow[r,"c"]  &\bar E \hquo G \arrow[r,"{\alpha}"] & M\hquo G 
\end{tikzcd}
\end{equation}
We have
$$
\alpha^*([\beta_G])=\tilde\gamma_G\;\;\;\;\text{ and }\;\;\;\; c^*(\tilde\gamma_G)=\gamma_G.
$$

\Lem{\label{lem:pull back along alpha of (0,Phi,beta)}
We have
$$
\alpha^*[(0,\Phi,\beta)] = [(0,\tilde\phi,\gamma)
].
$$
}
\Proof{
{Follows from} 
the first {equation} above and the corresponding descriptions in the total complexes.
}



\Pro{\label{pro:G-equiv section iff section of Borel iff beta-G=0 iff gamma-G=0}
The following are equivalent.
\begin{enumerate}
\item $i:NA\to E$ has a $G$-equivariant section.
\item $i:NA\to E\hquo G$ has a section.
\item $[\beta_G] \in H^2(M\hquo G,A)$ is zero.  
\item $[\gamma_G] \in H^2(\overline{E\hquo G},A)$ is zero.
\end{enumerate}
}
\Proof{
Equivalence of (1), (2) and (4) is by Lemma \ref{lem:equivariant section vs normal section} and Proposition  \ref{pro:section of j vs gamma_G}. Applying Proposition \ref{pro:equivalence section, zero cofib, zero fib} to the leftmost cofibration and the rightmost fibration, which is a partial group extension by Proposition \ref{pro:Borel construction as a semidirect product}, we obtain the equivalence of (3) and (4).
}

\Cor{\label{cor:beta-G nonzero then G-equiv contextuality}
If $[\beta_G]\neq 0$ then every 
$G$-equivariant  distribution $p:E\to D_R(NA)$ relative to $NA$ is $G$-equivariantly contextual.
}
\Proof{
Follows from Corollary \ref{cor:class nonzero implies equivariant contexuality} and Proposition \ref{pro:G-equiv section iff section of Borel iff beta-G=0 iff gamma-G=0}.
}

\Rem{\label{rem:equiv contextuality double complex}
{\rm
By Theorem \ref{thm:betaG zero iff there exists s iff equivariant section} we can show that $[\beta_G]\neq 0$ by working in the total complex. 
Part (2) of this corollary can be broken into two steps: We have $[\beta_G]=0$ if and only if  
\begin{enumerate}
\item there exists $s':M_1\to A$ such that $d^vs'=\beta$, and
\item there exists $s'':M_1\to A$ such that $d^vs''=0$ and $d^hs'' = \Phi-d^hs'$.
\end{enumerate}
Note that {$s=(0,s''+s')$} satisfies $ds=(0,\Phi,\beta)$. 
We can apply this observation to the case when $[\beta]=0$. Then {$s'$} as specified in (1) exists. Thus $[\beta_G]=0$ if and only if {$s''$} as specified in (2) exists. 
}
}

\subsection{{Example:} Dihedral group}

We will compare the cofibration involving the torus in Sections  \ref{sec:torus-with-involution} and \ref{sec:torus-rel-to-diag} to a partial group extension obtained from the trivial group extension:
\begin{equation}\label{dia:comparison dihedral}
\begin{tikzcd}[column sep=huge,row sep =large]
S^1 \arrow[d] \arrow[r] & N\ZZ_2 \arrow[d] \\
S^1\times S^1 \arrow[r] \arrow[d] & N(\ZZ_2\times \ZZ_2) \arrow[d] \\
\overline{S^1\times S^1} \arrow[r,"\alpha"] & N\ZZ_2
\end{tikzcd}
\end{equation}
The map $\alpha$ sends $x_i\mapsto 1$ where $i=0,1$.
We consider the {$G=\ZZ_2$} action on the middle spaces by swapping the coordinates.
Let us begin with a pseudo-section $\eta:N\ZZ_2\to N(\ZZ_2\times \ZZ_2)$ that sends $a\mapsto (0,a)$ in degree $1$. 
In fact, this is a section and hence
$\beta=0$. In Section \ref{sec:torus-rel-to-diag} we have seen that $[\gamma]=0$. Therefore $\alpha^*([\beta])=[\gamma]$ {(thus verifying Lemma \ref{lem:beta pullbacks to gamma})}. 
We compute $\Phi$ as follows:
$$
\begin{aligned}
\Phi_a(b) &= a\cdot \eta(-a\cdot b) \cdot \eta(b)^{-1} \\
&= \left\lbrace
\begin{array}{cc}
0   & a=0 \\
b    & a=1.
\end{array}
\right.
\end{aligned}
$$
Then we have
$$
\begin{aligned}
\alpha^*(\Phi)(a,x_i) &= \Phi(a,\alpha(x_i)) \\
&= \Phi(a,1) \\
&= \left\lbrace
\begin{array}{cc}
0    & a=0 \\
1    & a=1
\end{array}
\right. \\
&= \phi(a,1).
\end{aligned}
$$
where $\phi$ is given by Equation (\ref{eq:phi-torus}). Using $s$ in Equation (\ref{eq:s}) we observe that 
$$
\alpha^*(0,\Phi,0) = (0,\phi,0) = d(0,s) +(0,\phi,\gamma),
$$
and hence $\alpha^*[(0,\Phi,0)]=[(0,\phi,\gamma)]$ {(this verifies Lemma \ref{lem:pull back along alpha of (0,Phi,beta)})}.

We observe that $(\ZZ_2\times \ZZ_2)\rtimes \ZZ_2$, where the action is the swap, is isomorphic to the dihedral group $D_8$.
The action on the base space $N\ZZ_2$ is trivial. Then the Borel construction {of Diagram (\ref{dia:comparison dihedral})} gives 
$$
\begin{tikzcd}[column sep=huge,row sep =large]
S^1 \arrow[d] \arrow[r] & N\ZZ_2 \arrow[d] \\
(S^1\times S^1) \hquo G \arrow[r] \arrow[d] & ND_8 \arrow[d] \\
\overline{(S^1\times S^1)\hquo G} \arrow[r,"\alpha_G\,\circ\, c"] & N(\ZZ_2\times \ZZ_2)
\end{tikzcd}
$$
We can compute the {cocycle representing the} extension class from the pseudo-section $\eta_G:N(\ZZ_2\times \ZZ_2)\to ND_8$ defined by $\eta_G(c,d)=((0,c),d)$:
{\begin{align*}
\beta_G((c,d),(c',d')) &= \eta(c)\cdot (d\eta(c))\cdot \eta(c+(dc'))^{-1}\\
&= (0,c)+d\cdot(0,c')+(0,c+c')\\
&=
\left\{
\begin{array}{cc}
c' & d=1\\
0 & d=0.
\end{array}
\right. 
\end{align*}
Note that we use that the action on the base space is trivial.
}
Using Equation (\ref{eq:nabla2 formula}) we obtain
$$
\nabla_2(\beta_G) = (0,\Phi,0).
$$
When computing this we also used the isomorphism $ N(\ZZ_2\times \ZZ_2)\cong N\ZZ_2\hquo \ZZ_2$ that sends $((c,d),(c,d))\mapsto [(0,d,d'),(c,c')]$ in degree $2$. To illustrate this we compute the middle term in Equation (\ref{eq:nabla2 formula}):
$$
\begin{aligned}
\alpha'[(0,d),c] &= \beta_G[(0,d,0),(0,c)] - \beta_G[(0,0,d),(c,0)] \\
&= \beta_G[(0,d,0),(0,c)] \\
& = \left\lbrace
\begin{array}{cc}
c  & d=1\\
0 &  d=0
\end{array}
\right. \\
&= \Phi_d(c).
\end{aligned}
$$

\subsection{Quantum distributions}

{We recall some basic constructions from \cite[Section 6]{okay2022simplicial} to study simplicial distributions that come from quantum mechanics.}
Let $\hH$ denote a finite dimensional complex Hilbert space. We will write $\Pos(\hH)$ and $\Proj(\hH)$ for the set of positive operators and the subset of projectors. 
We define a functor $P_\hH:\catSet\to \catSet$:
\begin{itemize}
\item For a set $U$ the set $P_\hH(U)$ of projective measurements on $U$ consists of functions  $\Pi:U\to \Proj(\hH)$ {with finite support, i.e., $|\set{u\in U:\, \Pi(u)=\zero}|<\infty$,} such that $\sum_{u\in U} \Pi(u)=\one_\hH$.
\item Given a function $f:U\to V$ the function $P_\hH(f):P_\hH(U)\to P_\hH(V)$ is defined by
		$$
		\Pi \mapsto \left( v\mapsto \sum_{u\in f^{-1}(v)} \Pi(u) \right). 
		$$  
\end{itemize}
There is an analogous functor $Q_\hH$ involving $\Pos(\hH)$ instead of the projectors. Note that for $\hH=\CC$ this functor coincides with the distribution monad $D$ {introduced in Section \ref{sec:simplicial distributions}}.
In this paper we will focus on the projective version.
Given a simplicial set $X:\catDelta^\op\to \catSet$ we define $P_\hH(X):X\xrightarrow{X} \catSet \xrightarrow{P_\hH} \catSet$.

\Def{
A simplicial projective measurement on $(X,Y)$ is a simplicial set map $\Pi:X\to P_\hH(Y)$. We will write $\sProj(X,Y)$ for the set of simplicial projective measurements on $(X,Y)$. 
}

Let $\rho$ be a density operator, i.e., $\rho\in \Pos(\hH)$ and $\Tr(\rho)=1$. 
{B}y sending a simplicial measurement $\Pi:X\to P_\hH(Y)$ to the simplicial distribution $\rho_*(p):X\to D(Y)$ defined by 
$$
\sigma \mapsto \left(\,  \theta \mapsto \Tr(\rho\Pi_\sigma(\theta)) \,\right)
$$
{we obtain a simplicial set map
\begin{equation}\label{dia:simplicial Born}
\rho_*:P_\hH(Y) \to D(Y)
\end{equation}
{The trace $\Tr(\rho\Pi_\sigma(\theta))$ is interpreted 
as
the probability of obtaining the outcome $\theta$ and in physics literature this is known as the Born rule.}
The simplicial set map in (\ref{dia:simplicial Born}) {will be referred to as the {\it simplicial Born rule}}.
}
{Next we introduce a nerve space that has been successful in studying contextuality in the context of quantum theory; see \cite{chung2023simplicial} for applications to operator solutions of linear systems.}
\Def{\label{def:partial group N(Zd,K)}
{\rm
Let $d\geq 2$ be an integer and $K$ be a group. We define a partial group $N(\ZZ_d,K)$, a simplicial subset of the nerve space $NK$, whose $n$-simplices are given by
$$
N(\ZZ_d,K)_n=\set{ (k_1,k_2,\cdots,k_n):\, k_i^d=1_K,\;\; k_ik_j=k_jk_i\; \forall 1\leq i,j\leq n }.
$$ 
The property $k^d=1_K$ will be
 referred to as the $d$-torsion property.
}
}

Let $U(\hH)$ denote the group of unitary operators.

\Pro{(\!\cite[Proposition 6.3]{okay2022simplicial})
Sending an $n$-tuple $(A_1,A_2,\cdots,A_n)$ of pairwise commuting $d$-torsion unitary operators to the projective measurement $\Pi_A:\ZZ_d^n\to \Proj(\hH)$ obtained by simultaneously diagonalizing the operators gives an isomorphism of simplicial sets
$$
\sd:N(\ZZ_d,U(\hH)) \to P_\hH(N\ZZ_d)
$$
}

We will identify these two simplicial sets. After this identification the simplicial Born rule becomes
$$
\rho_*:N(\ZZ_d,U(\hH)) \to D(N\ZZ_d)
$$ 
Now, consider a group $K$
with a central element $J$ or order $d$. 
We have a
  central group extension
\begin{equation}\label{eq:central group extension J}
1 \to \Span{J} \to K \xrightarrow{\epsilon} \bar K\to 1
\end{equation}

We write $\bar g$ for the image of $g$ under the quotient map.

\Def{\label{def:partial group bar N(Zd,K)}
{\rm
Let $d\geq 2$ be an integer and $K$ be a group with a central element $J$ of order $d$. We define a partial group $\bar N(\ZZ_d,K)$, a simplicial subset of the nerve space $N\bar K$, whose $n$-simplices are given by
$$
\bar N(\ZZ_d,K)_n=\set{ (\bar k_1,\bar k_2,\cdots,\bar k_n)\in \bar K^n:\, k_i^d=1_K,\;\; k_ik_j=k_jk_i\; \forall 1\leq i,j\leq n }.
$$  
}
}

The central extension in (\ref{eq:central group extension J}) gives a diagram as in (\ref{dia:group extension as partial group}). Puling back this diagram  along the inclusion $\bar N(\ZZ_d,K)\to N\bar K$ we obtain the following result. We will still write $\beta$ for the pull-back of the cocycle $\beta$ along this inclusion.

\Pro{\label{pro:central extension nerve spaces}
We have a central extension of partial groups
\begin{equation}\label{eq:partial group extension nerve spaces}
N\ZZ_d \xrightarrow{i} N(\ZZ_d,K) \xrightarrow{q} \bar N(\ZZ_d,K)
\end{equation} 
where the total space can 
be described as the twisted product $N\ZZ_d \times_\beta \bar N(\ZZ_d,K)$.
}
 
There is a way to obtain a simplicial distribution on the total space relative to the fiber that uses representation theory {of  groups}. 
Consider a unitary representation $\psi:K\to U(\hH)$, i.e., a group homomorphism, such that 
$$
\psi(J) = e^{2\pi i/d} \one.
$$
Then given a quantum state $\rho\in \Den(\hH)$ we can construct a commutative diagram
$$
\begin{tikzcd}[column sep=huge,row sep =large]
N\ZZ_d \arrow[d] \arrow[r,equal] & N\ZZ_d \arrow[d] \arrow[r,equal] &  N\ZZ_d \arrow[d,"\delta"]  \\
N(\ZZ_d,K) \arrow[d] \arrow[r,"\psi_*"] & N(\ZZ_d,U(\hH)) \arrow[d] \arrow[r,"{\rho_*}"] & D(N\ZZ_d) \\
\bar N(\ZZ_d,K) \arrow[r] & \bar N(\ZZ_d,U(\hH)) &
\end{tikzcd}
$$
The composite $\rho_*^\psi= \rho_*\circ\psi_*$ is a simplicial distribution relative to $N\ZZ_d$.

\Def{\label{def:quantum contextuality}
A quantum state $\rho\in \Den(\hH)$ is called (non-)contextual with respect to {$\psi:K\to U(\hH)$} if the simplicial distribution 
\begin{equation}\label{eq:rho psi}
\rho_*^\psi: N(\ZZ_d,K) \to D(\ZZ_d)
\end{equation}
is (non-)contextual.
}


\Cor{\label{cor:beta non zero rho contextual}
If $[\beta]\neq 0$ then $\rho_*^\psi$ is contextual for all  $\rho\in \Den(\hH)$.
}
\Proof{
{Follows from Corollary \ref{cor:beta nonzero then contextual} since the simplicial distribution $\rho_*^\psi$ is relative to $NA$.}
}

\Rem{
{\rm
In practice we can assume that $\psi$ is injective so that $K$ can be identified with a subgroup of $U(\hH)$. Note that we can always replace $K$ with the image of $\psi$. 
$K$ is called {\it contextual} if $\rho_*^\psi$ is contextual for all $\rho\in \Den(\hH)$.
This notion of contextuality is usually referred to as state-independent contextuality and a topological perspective is first developed in \cite{Coho}. {In this case we say $\rho$ is contextual with respect to $K$ instead of the homomorphism $\psi$.}
}
}

Next, we turn to the equivariant case.
Let us assume that $\psi:K\to U(\hH)$ is injective so that we can identify $K$ as a subgroup of the unitary group.
Under this assumption we drop $\psi$ from notation, e.g., we simply write $\rho_*$ for the map in (\ref{eq:rho psi}).
We will consider partial group actions on $N(\ZZ_d,K)$ in the sense of Section \ref{sec:Group action on partial groups}. In practice such actions arise from a subgroup $G$ of the normalizer $N(K)$ of $K$ in the unitary group $U(\hH)$. For the rest of this section we will restrict to actions arising in this way.

\Def{\label{def:equivariant contextuality rho}
A quantum state $\rho$ is called $G$-equivariant if $U\rho U^\dagger =\rho$ for all $U\in G$. 
A $G$-equivariant quantum state is $G$-equivariantly (non-)contextual if $\rho_*$ is $G$-equivariantly (non-)contextual. 
}

\Cor{\label{cor:beta-G non zero rho contextual}
If $[\beta_G]\neq 0$ then  $\rho_*$ is $G$-equivariantly contextual for all $G$-equivariant quantum states $\rho$. 
}
\Proof{
{Follows from Corollary \ref{cor:beta-G nonzero then G-equiv contextuality} since the simplicial distribution $\rho_*^\psi$ is $G$-equivariant and relative to $NA$.}
}

In practice, we will use Remark \ref{rem:equiv contextuality double complex} to show that $[\beta_G]\neq 0$. 
{The situation in Corollary \ref{cor:beta-G non zero rho contextual} is referred to as state-independent contextuality in the physics literature \cite{Coho}.}

\subsection{Example: Mermin star}
\label{sec:Mermin star}

In this section we will construct an example that is of practical interest in {quantum foundations and computing}. {Our example is a topological representation of the well-known Mermin star construction introduced in \cite{mermin1993hidden}. We improve on the earlier topological approach of \cite{Coho} by help of the simplicial approach developed in \cite{okay2022simplicial}.}
 
{Let} $\hH$ {denote} the $n$-fold tensor product $(\CC^2)^{\otimes n}$.
{The Pauli $X$ and $Z$ matrices are given by}
$$
X = \left(\begin{matrix}
0 & 1 \\
1 & 0
\end{matrix} \right) \;\;\;\;
Z = \left(\begin{matrix}
1 & 0 \\
0 & -1
\end{matrix} \right).
$$
{The Pauli group} $P_n$ {is} the subgroup of $U(\hH)$ generated by operators of the form {
$$
T_a = i^{a_z\cdot a_x} Z(a_z) X(a_x)
$$
{where $i$ is the imaginary unit,}
 $a=(a_z,a_x)\in \ZZ_2^n\times \ZZ_2^n$, {$a_z\cdot a_x=\sum_j (a_z)_j(a_x)_j$}, and 
$$
Z(a_z) = Z^{(a_z)_1}\otimes \cdots \otimes Z^{(a_z)_n}  \;\;\;\;\;\; X(a_x) = X^{(a_x)_1}\otimes \cdots \otimes X^{(a_x)_n}.
$$
Note that $T_a^2=\one$ for all $a\in \ZZ_2^{2n}$.
The Pauli group fits in a central extension
$$
0\to \ZZ_4\xrightarrow{1\mapsto i\one}  P_n \xrightarrow{\epsilon} \ZZ_2^{2n} \to 0
$$
A set-theoretic {section} of $\epsilon$ can be defined by
\begin{equation}\label{eq:eta Ta}
\eta(a) =  T_a.
\end{equation}
}

We will study the partial group extension
$$
N\ZZ_2 \to N(\ZZ_2,P_n) \xrightarrow{\pi} \bar N(\ZZ_2,P_n)
$$

  {
Now, the set-theoretic section in Equation (\ref{eq:eta Ta}) can \iis{be} used to define a pseudo-section of $\pi$, which we still denote by $\eta$. Using $\eta$ and Equation (\ref{eq:cocycle from eta}) we can compute
$$
\begin{aligned}
(-1)^{\beta(a,b)} \one &= \eta(a)\cdot \eta(b)\cdot \eta(a+b)^{-1} \\
&= T_a T_b (T_{a+b})^{-1} \\
&= (i^{a_z\cdot a_x} Z(a_z)X(a_x))(i^{b_z\cdot b_x} Z(b_z)X(b_x))(i^{(a+b)_z\cdot (a+b)_x} Z((a+b)_z)X((a+b)_x)) \\
&= i^{a_z\cdot a_x+b_z\cdot b_x+(a+b)_z\cdot (a+b)_x} (-1)^{b_z\cdot a_x+(a+b)_z\cdot (a+b)_x} \\
&= i^{b_z\cdot a_x - a_z\cdot b_x} \\
&= (-1)^{(b_z\cdot a_x - a_z\cdot b_x)/2} 
\end{aligned}
$$
where we used the fact that $b_z\cdot a_x - a_z\cdot b_x$ is divisible by $2$ since $\omega(a,b)=0$. Therefore 
\begin{equation}\label{eq:beta for Pn}
\beta(a,b)=(b_z\cdot a_x - a_z\cdot b_x)/2.
\end{equation}
 Similarly we can compute $\Phi$ using Equation (\ref{eq:phi from eta}).  
}
 
\begin{figure}[h!]
\centering
\begin{subfigure}{.49\textwidth}
  \centering
  \includegraphics[width=.85\linewidth]{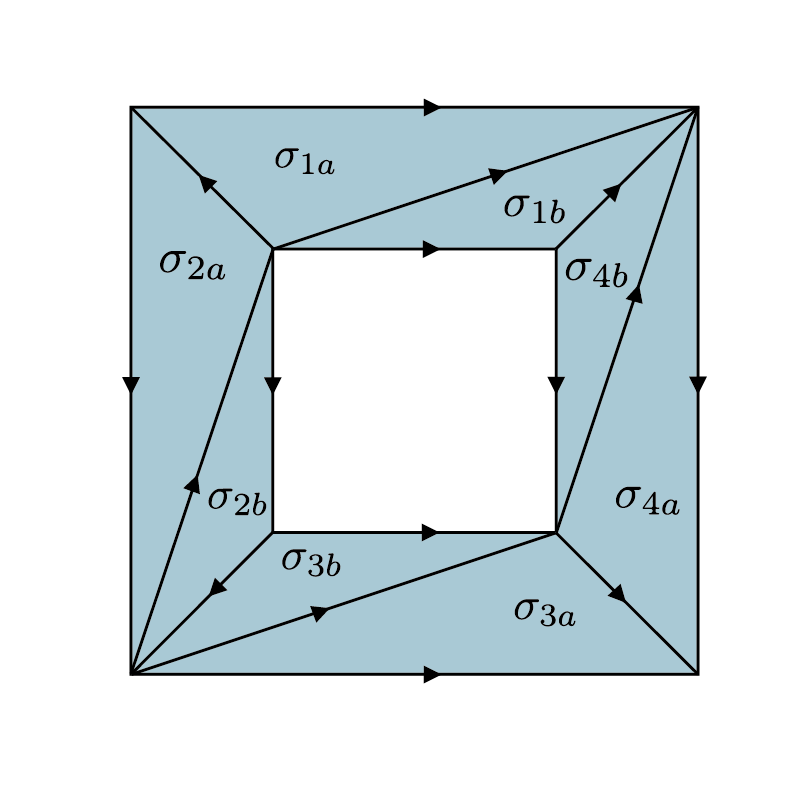}
  \caption{}
  \label{fig:mermin}
\end{subfigure}%
\begin{subfigure}{.49\textwidth}
  \centering
  \includegraphics[width=.85\linewidth]{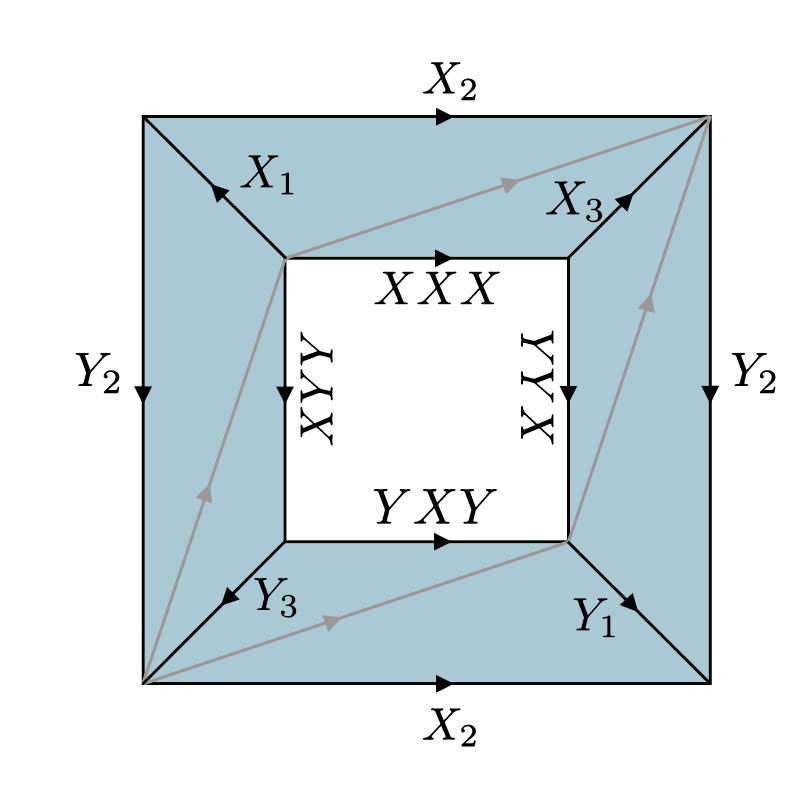}
  \caption{}
  \label{fig:mermin-operators}
\end{subfigure}

\caption{(a)The Mermin star represented by a reduced simplicial set $M$ consisting of eight $2$-simplices $\sigma_{ai},\sigma_{bi}$ where $i=1,2,3,4$. The right (top) and left (bottom) edges are identified.  
(b) Edges are labeled by elements of $P_3$ modulo the subgroup $\Span{-\one}$. {For simplicity we omit the tensor product from notation, e.g., $XXX$ stands for $X\otimes X\otimes X$, and indicate the qubit on which the operator is acting by a subscript, e.g., $X_1$ stands for $X\otimes \one\otimes \one$.}
}
\label{fig:mermin-scenario-and-os}
\end{figure} 

We begin by constructing a smaller partial group extension and a cofibration that comes with a comparison map. For the partial group extension first we consider the $2$-skeleton $N\ZZ_2^{(2)}$ and the partial monoid $M$ given by the space in Figure (\ref{fig:mermin}). Let $\bar \Pi:M\to \bar N(\ZZ_2,P_n)$ denote the inclusion described by Figure (\ref{fig:mermin-operators}) and $E$ denote the $2$-skeleton of the pull-back of $\pi$ along this maps.
{The space $E$ consists of the non-degenerate $2$-simplices given by $\sigma_{ai}(cd)$ and $\sigma_{bi}(cd)$ where $c,d\in \ZZ_2$, i.e., there are four copies of triangles above a given triangle of $M$. Instead of providing a picture of $E$ w}e will consider a simplicial subset $X\subset E$ which is described in Figure (\ref{fig:small-mermin}). {This portion focuses on the restriction of $E$ over the two triangles $\sigma_{1b}$ and $\sigma_{4b}$ of $M$.}  
Using all these spaces we obtain a commutative diagram
$$
\begin{tikzcd}[column sep=huge,row sep =large]
  S^1 \arrow[d,"i"] \arrow[r,hook] & N\ZZ_2^{(2)} \arrow[d] \arrow[r,hook] & N\ZZ_2 \arrow[d] \\
 X \arrow[d,"q"] \arrow[r,"\alpha",hook] & E \arrow[d,"\pi"] \arrow[r,hook,"\Pi"] & N(\ZZ_2,P_3) \arrow[d] \\
 \bar X \arrow[r,"\bar \alpha"]   & M \arrow[r,hook,"\bar \Pi"] & \bar N(\ZZ_2,P_3)
\end{tikzcd}
$$ 
where $\Pi|_X = \Pi\circ \alpha$ is described in Figure (\ref{fig:small-mermin-operators}). Under $i$ the circle $S^1$ maps to the common edge {$x$} between {$\sigma_2$ and $\sigma_3$}.

 \begin{figure}[h!]
\centering
\begin{subfigure}{.49\textwidth}
  \centering
  \includegraphics[width=.7\linewidth]{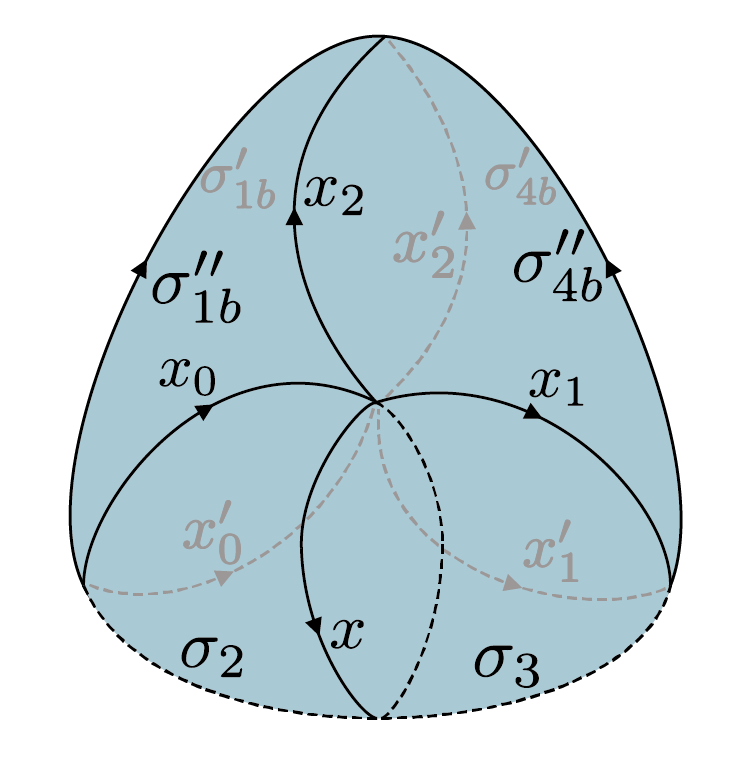}
  \caption{}
  \label{fig:small-mermin}
\end{subfigure}%
\begin{subfigure}{.49\textwidth}
  \centering
  \includegraphics[width=.7\linewidth]{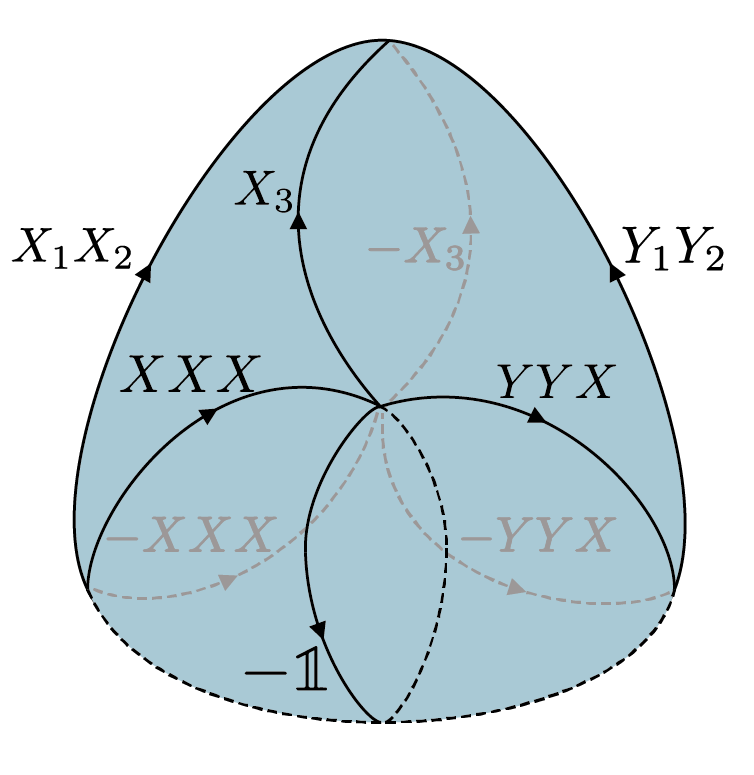}
  \caption{}
  \label{fig:small-mermin-operators}
\end{subfigure}

\caption{(a) The reduced simplicial set $X$ consisting of six $2$-simplices {$\sigma'_{ib},\sigma''_{ib},\sigma_{2},\sigma_3$} where $i=1,4$.
(b) Edges are labeled by elements of $P_3$. This specifies a simplicial set map $\Pi|_X:X\to N(\ZZ_2,P_3)$. 
}
\label{fig:mermin-scenario-and-os}
\end{figure} 
 
Next we describe the symmetry group $G$. {This example is a modified version of the one discussed in the context of measurement-based quantum computing; see \cite{raussendorf2016cohomological}.}
Let us define the following unitary operator
$$
A = \frac{X + Y}{\sqrt{2}}.
$$ 
Note that $A^2=\one$ and we have
$$
A X A = Y,\;\;\;\; AYA = X, \;\;\;\; AZA = -Z.
$$ 
Therefore $A$ is in the normalizer of $P_1\subset U(\CC^2)$. Consider the symmetry group
$$
G = \Span{A\otimes A \otimes Y, A \otimes Y\otimes A, Y\otimes A \otimes A}.
$$ 
Note that $G$ acts on $N(\ZZ_2,P_3)$ and furthermore this action restricts to an action on $E${, which is regarded as a simplicial subset via the inclusion $\Pi$}. Consider the following subgroup
$$ 
H  = \Span{V} \subset G.
$$
where $V= A\otimes A \otimes Y$.
We have
$$
\begin{aligned}
V (X\otimes X\otimes \one) V &= Y\otimes Y\otimes \one \\
V (X\otimes X\otimes X) V &= -Y\otimes Y\otimes X \\
V (\one\otimes \one\otimes X) V &= -\one\otimes \one\otimes X.  
\end{aligned}
$$
Hence  $H$ acts on the smaller simplicial subset $X$ {(embedded via $\Pi|_X$)}.  

Finally we consider a $G$-equivariant quantum state $\rho\in \Den(\hH)$ given by
{
\begin{align}
  \begin{split}
 \rho = \frac{1}{4}(\one + \one\otimes Z\otimes Z &+ Z\otimes \one\otimes Z+Z\otimes Z\otimes \one \\
 &+ X\otimes X\otimes X - X\otimes Y\otimes Y - Y\otimes X\otimes Y - Y\otimes Y\otimes X).
  \end{split}
\end{align}
Note that $\rho$ can be written as the outer product $vv^\dagger$ where $v=(e_0\otimes e_0\otimes e_0 + e_1\otimes e_1\otimes e_1)/\sqrt{2}$ and $\set{e_0,e_1}$ is the canonical basis of $\CC^2$.
}
In physics literature the state $v$ is called the {Greenberger--Horn--Zeilinger (GHZ) state \cite{nielsen2010quantum}}. Observe that
$$
U \rho U^\dagger = \rho
$$
for all $U\in G$. Therefore we obtain a $G$-equivariant distribution using the simplicial Born rule {(see (\ref{dia:simplicial Born}))}
$$
\rho_*: N(\ZZ_2,P_3) \to D(N\ZZ_2)
$$
which we can restrict to $E$ and $X$ to obtain the $G$-equivariant distributions $\rho_*|_E$ and $\rho_*|_X$, respectively.

First we begin with $X$ and the $H$-equivariant distribution $p_X=\rho_*|_X$.
It turns out that $X$ does not admit any $H$-equivariant deterministic distribution. To see this note that such a distribution $d$ is determined by the assignment {$d_{d_2\sigma''_{1b}} =\delta^a$} and {$d_{d_0\sigma''_{1b}}=\delta^b$}. Then the $2$-simplex {$\sigma_2$} forces that {$d_{d_2\sigma'_{1b}}=\delta^{1-a}$} since one of the edges is given by $\delta^1$.
In addition, $H$-equivariance implies that {$d_{d_0\sigma'_{1b}}=\delta^b$}. Then we obtain a contradiction
$$
\delta^{a+b} = d_{d_1{\sigma''_{1b}}} = d_{d_1{\sigma'_{1b}}} = \delta^{1+a+b}.
$$  
By Corollary \ref{cor:beta-G non zero rho contextual} this implies that $[\gamma_G]\neq 0$.
On the other hand, we observe that $[\gamma]=0$ since the cocycle $\gamma$ is non-zero on two of the non-degenerate $2$-simplices {$\sigma_2$ and $\sigma_3$}. {Let us define
$$
s(x') = \left\lbrace
\begin{array}{ll}
1 & x' \in \set{x_i:\, i=0,1,2}\\
0 & \text{otherwise.}
\end{array}
\right.
$$
Then using $\phi=-d^hs$ we obtain
$$
\phi(g,x') = \left\lbrace
\begin{array}{ll}
1 & g=1,\; x' \in \set{x_i,x_i':\,i=0,1,2}\\
0 & \text{otherwise.}
\end{array}
\right.
$$ 
Let us consider the partial group extension $E$. We can compute $\beta$, as we did for the general case; see Equation (\ref{eq:beta for Pn}). Since in each triangle the product of operators on the boundary is $\one$ the cocycle $\beta$ is zero. Let us compute $\Phi$ for the group element $V$:
$$
(-1)^{\Phi_g(x')} = g\cdot \eta(g\cdot x')  \eta(x') = \left\lbrace
\begin{array}{ll}
-1 & g=V, \;x' \in \set{x_i,x_i':\, i=0,1,2} \\
1 & \text{otherwise.}
\end{array}
\right.
$$ 
Therefore $\phi = \bar\alpha^*(\Phi)$.  }
  

\bibliography{bib.bib}

\begin{thebibliography}{10}

\bibitem{abramsky2011sheaf}
S.~Abramsky and A.~Brandenburger, ``The sheaf-theoretic structure of
  non-locality and contextuality,'' {\em New Journal of Physics}, vol.~13,
  no.~11, p.~113036, 2011.
\newblock doi:
  \href{https://doi.org/10.1088/1367-2630/13/11/113036}{10.1088/1367-2630/13/11/113036}.

\bibitem{Coho}
C.~Okay, S.~Roberts, S.~D. Bartlett, and R.~Raussendorf, ``Topological proofs
  of contextuality in quantum mechanics,'' {\em Quantum Information \&
  Computation}, vol.~17, no.~13-14, pp.~1135--1166, 2017.
\newblock doi:
  \href{https://doi.org/10.26421/QIC17.13-14-5}{10.26421/QIC17.13-14-5}.

\bibitem{okay2022simplicial}
C.~Okay, A.~Kharoof, and S.~Ipek, ``Simplicial quantum contextuality,'' {\em
  Quantum}, vol.~7, 2023.
\newblock doi:
  \href{https://doi.org/10.22331/q-2023-05-22-1009}{10.22331/q-2023-05-22-1009}.

\bibitem{raussendorf2016cohomological}
R.~Raussendorf, ``Cohomological framework for contextual quantum
  computations,'' {\em Quantum Information and Computation}, vol.~19,
  no.~13{\&}14, pp.~1141--1170, 2019.
\newblock doi:
  \href{https://doi.org/10.26421/QIC19.13-14-4}{10.26421/QIC19.13-14-4}.

\bibitem{Fraction2018}
C.~Okay, E.~Tyhurst, and R.~Raussendorf, ``The cohomological and the
  resource-theoretic perspective on quantum contextuality: common ground
  through the contextual fraction,'' {\em Quantum Information and Computation},
  vol.~18, no.~15{\&}16, pp.~1272--1294, 2018.
\newblock doi:
  \href{https://doi.org/10.26421/QIC18.15-16-2}{10.26421/QIC18.15-16-2}.

\bibitem{raussendorf2023role}
R.~Raussendorf, C.~Okay, M.~Zurel, and P.~Feldmann, ``The role of cohomology in
  quantum computation with magic states,'' {\em Quantum}, vol.~7, p.~979, 2023.
\newblock doi:
  \href{https://doi.org/10.22331/q-2023-04-13-979}{10.22331/q-2023-04-13-979}.

\bibitem{broto2021extension}
C.~Broto and A.~Gonzalez, ``An extension theory for partial groups,'' {\em
  arXiv preprint arXiv:2105.03457}, 2021.
\newblock doi:
  \href{https://doi.org/10.48550/arXiv.1507.04392}{10.48550/arXiv.1507.04392}.

\bibitem{adem2015classifying}
A.~Adem and J.~G{\'o}mez, ``A classifying space for commutativity in lie
  groups,'' {\em Algebraic \& Geometric Topology}, vol.~15, no.~1,
  pp.~493--535, 2015.
\newblock doi:
  \href{https://doi.org/10.2140/agt.2015.15.493}{10.2140/agt.2015.15.493}.

\bibitem{antolin2020classifying}
O.~Antol{\'\i}n-Camarena, S.~P. Gritschacher, and B.~Villarreal, ``Classifying
  spaces for commutativity of low-dimensional lie groups,'' in {\em
  Mathematical Proceedings of the Cambridge Philosophical Society}, vol.~169,
  pp.~433--478, Cambridge University Press, 2020.
\newblock doi:
  \href{https://doi.org/10.1017/S0305004119000240}{10.1017/S0305004119000240}.

\bibitem{okay2021classifying}
C.~Okay and D.~Sheinbaum, ``Classifying space for quantum contextuality,'' in
  {\em Annales Henri Poincar{\'e}}, vol.~22, pp.~529--562, Springer, 2021.
\newblock doi:
  \href{https://doi.org/10.1007/s00023-020-00993-3}{10.1007/s00023-020-00993-3}.

\bibitem{mermin1993hidden}
N.~D. Mermin, ``Hidden variables and the two theorems of {J}ohn {B}ell,'' {\em
  Reviews of Modern Physics}, vol.~65, no.~3, p.~803, 1993.
\newblock doi:
  \href{https://link.aps.org/doi/10.1103/RevModPhys.65.803}{10.1103/RevModPhys.65.803}.

\bibitem{kharoof2022simplicial}
A.~Kharoof and C.~Okay, ``Simplicial distributions, convex categories and
  contextuality,'' {\em preprint arXiv:2211.00571}, 2022.
\newblock doi:
  \href{https://doi.org/10.48550/arXiv.2211.00571}{10.48550/arXiv.2211.00571}.

\bibitem{kharoof2023topological}
A.~Kharoof, S.~Ipek, and C.~Okay, ``Topological methods for studying
  contextuality: N-cycle scenarios and beyond,'' {\em Entropy}, vol.~25, no.~8,
  2023.
\newblock doi: \href{https://doi.org/10.3390/e25081127}{10.3390/e25081127}.

\bibitem{joyal2008theory}
A.~Joyal, ``The theory of quasi-categories and its applications,'' {\em
  Quaderns}, vol.~45, no.~2, pp.~151--496, 2008.

\bibitem{weibel1995introduction}
C.~A. Weibel, {\em An introduction to homological algebra}.
\newblock No.~38, Cambridge university press, 1995.
\newblock doi: \href{https://doi.org/10.1007/b98977}{10.1007/b98977}.

\bibitem{chung2023simplicial}
H.~Y. Chung, C.~Okay, and I.~Sikora, ``Simplicial techniques for operator
  solutions of linear constraint systems,'' {\em preprint arXiv:2305.07974},
  2023.
\newblock doi:
  \href{https://doi.org/10.48550/arXiv.2305.07974}{10.48550/arXiv.2305.07974}.

\bibitem{adem2013cohomology}
A.~Adem and R.~J. Milgram, {\em Cohomology of finite groups}, vol.~309.
\newblock Springer Science \& Business Media, 2013.
\newblock doi:
  \href{https://doi.org/10.1007/978-3-662-06280-7}{10.1007/978-3-662-06280-7}.

\bibitem{mccleary2001user}
J.~McCleary, {\em A user's guide to spectral sequences}.
\newblock No.~58, Cambridge University Press, 2001.
\newblock doi:
  \href{https://doi.org/10.1017/CBO9780511626289}{10.1017/CBO9780511626289}.

\bibitem{nielsen2010quantum}
M.~A. Nielsen and I.~L. Chuang, {\em Quantum computation and quantum
  information}.
\newblock Cambridge university press, 2010.
\newblock doi:
  \href{https://doi.org/10.1017/CBO9780511976667}{10.1017/CBO9780511976667}.

\end{thebibliography}
\bibliographystyle{ieeetr}

\end{document}